\documentclass[11pt,english]{article}

\usepackage[T1]{fontenc}
\usepackage{babel}
\usepackage{verbatim}
\usepackage{bm}
\usepackage{amsmath,amstext,amsthm,amssymb}
\usepackage{float}
\usepackage{xspace}
\usepackage{graphicx}
\usepackage{mathtools}
\usepackage{subcaption}
\usepackage{wrapfig}

\usepackage[dvipsnames]{xcolor}
\definecolor{DarkGray}{rgb}{0.66, 0.66, 0.66}
\definecolor{DarkPowderBlue}{rgb}{0.0, 0.2, 0.6}
\definecolor{fluorescentyellow}{rgb}{0.8, 1.0, 0.0}

\usepackage{tikz-qtree,tikz-qtree-compat}
\usepackage{tikz}

\usepackage[unicode=true,
bookmarks=false,
breaklinks=false,pdfborder={0 0 1},backref=none,colorlinks=true,allcolors=blue]
{hyperref}

\usepackage{cleveref}
\usepackage{thmtools}
\usepackage{thm-restate}
\usepackage{bbm}
\usepackage[bottom]{footmisc}
\usepackage[font={small,it}]{caption}

\usepackage[ruled,vlined,linesnumbered,algonl]{algorithm2e}
\SetEndCharOfAlgoLine{}
\SetKwComment{Comment}{\footnotesize$\triangleright$\ }{}

\SetCommentSty{mycommfont}

\Crefname{algocf}{Algorithm}{Algorithms}
\crefname{algocfline}{line}{lines}
\Crefname{invariant}{Invariant}{Invariants}
\Crefname{claim}{Claim}{Claims}
\Crefname{subclaim}{Subclaim}{Subclaims}

\usepackage{enumitem}
\usepackage{fullpage}
\usepackage{nicefrac}
\usepackage[compact]{titlesec}

\usepackage{setspace}

\makeatletter
\allowdisplaybreaks
\makeatother

\theoremstyle{plain}
\newtheorem{theorem}{Theorem}
\newtheorem{lemma}[theorem]{Lemma}

\newtheorem{obs}[theorem]{Observation}

\newtheorem{prop}[theorem]{Proposition}
\newtheorem{cor}[theorem]{Corollary}
\theoremstyle{plain}
\theoremstyle{definition}
\newtheorem{defn}[theorem]{\protect\definitionname}
\theoremstyle{plain}
\theoremstyle{remark}

\ifx\proof\undefined
\newenvironment{proof}[1][\protect\proofname]{\par
	\normalfont\topsep6\p@\@plus6\p@\relax
	\trivlist
	\itemindent\parindent
	\item[\hskip\labelsep\scshape #1]\ignorespaces
}{%
	\endtrivlist\@endpefalse
}
\providecommand{\proofname}{Proof}
\fi

\usepackage{babel}

\newcommand*\samethanks[1][\value{footnote}]{\footnotemark[#1]}

\usepackage[textsize=tiny,textwidth=2cm,color=green!50!gray,obeyFinal]{todonotes} %

\usepackage{soul}
\sethlcolor{fluorescentyellow}

\makeatother

\providecommand{\claimname}{Claim}
\providecommand{\definitionname}{Definition}

\newcommand{\nf}{\nicefrac}

\newcommand{\EE}{\mathbb{E}}

\newcommand{\RR}{\mathbb{R}}

\newcommand{\ones}{\mathbbm{1}}

\newcommand{\cA}{\mathcal{A}}
\newcommand{\cB}{\mathcal{B}}
\newcommand{\cC}{\mathcal{C}}

\newcommand{\cF}{\mathcal{F}}

\newcommand{\cI}{\mathcal{I}}

\newcommand{\cK}{\mathcal{K}}

\newcommand{\cS}{\mathcal{S}}

\newcommand{\cU}{\mathcal{U}}

\newcommand{\cX}{\mathcal{X}}

\newcommand{\cost}{\operatorname{cost}}

\newcommand\all{\operatorname{all}}

\newcommand{\eps}{\varepsilon}
\newcommand{\sse}{\subseteq}
\newcommand{\btau}{\pmb{\tau}}
\newcommand{\bmu}{\pmb{\mu}}

\newcommand{\be}{{\bm e}}
\newcommand{\bp}{{\bm p}}
\newcommand{\bq}{{\bm q}}
\newcommand{\bx}{{\bm x}}
\newcommand{\by}{{\bm y}}

\newcommand{\RT}{Random Thresholds\xspace}

\newcommand{\intervals}{\mathcal{I}}

\newcommand{\cuts}{\cK}

\newcommand{\Univ}{U}
\newcommand{\Set}{V}

\newcommand{\eat}[1]{}

\newcommand{\kmed}{$k$-medians\xspace}
\newcommand{\kmeans}{$k$-means\xspace}
\newcommand{\supp}{\operatorname{supp}}
\newcommand{\poly}{\operatorname{poly}}
\newcommand{\opt}{\ensuremath{\text{OPT}}\xspace}

\begin{document}

\title{The Price of Explainability for Clustering\thanks{Part of this work was done while visiting the \emph{Data-Driven Decision
Processes} semester program at the Simons Institute for the Theory of Computing.}}

\author{Anupam Gupta\thanks{Carnegie Mellon University. \texttt{anupamg@andrew.cmu.edu,
    mpittu@andrew.cmu.edu, rachely@andrew.cmu.edu}} \and Madhusudhan
  Reddy Pittu\samethanks[2]
  \and 
Ola Svensson\thanks{EPFL, Lausanne. \texttt{ola.svensson@epfl.ch}}
\and Rachel Yuan\samethanks[2]}

\maketitle
\begin{abstract}
  Given a set of points in $d$-dimensional space, an explainable
  clustering is one where the clusters are specified by a
  tree of axis-aligned threshold cuts. Dasgupta et al.\ (ICML
  2020) posed the question of the \emph{price of explainability}: the
  worst-case ratio between the cost of the best explainable
  clusterings to that of the best clusterings. 

  We show that the price of explainability for $k$-medians is at most
  $1+H_{k-1}$; in fact, we show that the popular Random Thresholds
  algorithm has \emph{exactly} this price of explainability, matching
  the known lower bound constructions. We complement our tight
  analysis of this particular algorithm by constructing instances
  where the price of explainability (using \emph{any} algorithm) is at
  least $(1-o(1)) \ln k$, showing that our result is best possible, up
  to lower-order terms.
  We  also improve the price of explainability for the
  $k$-means problem  to $O(k \ln \ln k)$ from the
  previous $O(k \ln k)$, considerably closing the gap to the lower
  bounds of $\Omega(k)$. 
  
  Finally, we study  
  the algorithmic question of finding the best explainable clustering:  We show that explainable $k$-medians and $k$-means cannot be approximated  better than  $O(\ln k)$,
  under standard complexity-theoretic conjectures. This essentially settles the approximability of explainable $k$-medians and leaves open the intriguing possibility to get significantly better approximation algorithms for $k$-means than its price of explainability.

\end{abstract}

\thispagestyle{empty}
\newpage{}
\setcounter{page}{1}

\section{Introduction}
\label{sec:introduction}

Clustering is a central topic in optimization, machine learning, and algorithm design, with $k$-medians and $k$-means being two of the most prominent examples.  
In recent years, mainly motivated by the impressive but still mysterious advances in machine learning,  there has been an increased interest in the transparency and in the explainability of solutions. 
In the context of clustering, this was formalized in a highly influential paper by Dasgupta et al.~\cite{DFMR20}.

To motivate the concept of explainability, consider the task of clustering $n$ points in $\mathbb{R}^d$ into $k$ clusters. If we solve $k$-means, the clusters are in general given by a Voronoi diagram where each cluster/cell is defined by the intersection of hyperplanes. %
Each cluster may be defined using up to $k-1$ hyperplanes, each one of them possibly depending on all $d$ dimensions with arbitrary coefficients. 
Since the dimensions typically correspond to features (e.g., ``age'', ``weight'', and ``height'' are natural features in a dataset of people), arbitrary linear combinations of these features may be difficult to interpret. 
To achieve more explainable solutions, we may need to restrict our algorithms to find clusters with simpler descriptions.

The model in~\cite{DFMR20} achieves explainability in an elegant way resembling the classical notion of decision trees in theoretical computer science. 
Specifically, a clustering is called \emph{explainable} if it is given by a decision tree, where each internal node splits data points with a \emph{threshold cut} in a single dimension (feature), and each of the $k$ leaves corresponds to a unique cluster. This leads to more explainable solutions already in two dimensions (see, e.g., \Cref{fig:optimal_vs_tree}); %
the benefit is even more clear in higher dimensions. Indeed, the binary tree structure gives an easy sequential procedure for classifying points, and since each threshold cut is axis-aligned, there is no linear combinations of features. Moreover, the total number of dimensions/features used to describe the clustering is at most $k-1$, independent of $d$, which is attractive for high-dimensional data\footnote{We remark that dimensionality reduction for $k$-median and $k$-means show that one can reduce the dimension of the data points to $O(\log(k)/\epsilon^2)$~\cite{BecchettiBC0S19,MakarychevMR19}. However, those techniques take arbitrary linear combinations of the original dimensions and therefore destroy explainability.}. %

\definecolor{cluster0color}{RGB}{162,201,221}
\definecolor{cluster1color}{RGB}{50,158,43}
\definecolor{cluster2color}{RGB}{233,176,102}
\definecolor{cluster3color}{RGB}{106,61,154}
\definecolor{cluster4color}{RGB}{165,83,37}

\begin{figure*}[t]
    \centering
    \subfloat[Optimal $5$-means clusters]{
         \includegraphics[width=.33\textwidth]{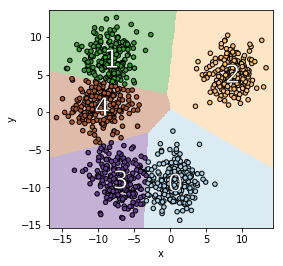}
    \label{fig:optimal_clusters}}
    \hfill
     \subfloat[Tree based $5$-means clusters]{
         \includegraphics[width=.33\textwidth]{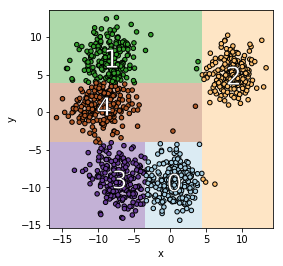}
      \label{fig:tree_clusters}
     }
     \hfill
     \subfloat[Threshold tree]{
            \resizebox{.25\textwidth}{!}{%
            \begin{tikzpicture}
            [inner/.style={shape=rectangle, rounded corners, draw, align=center, top color=white, bottom color=gray!40, scale=.85},
            leaf/.style={shape=rectangle, rounded corners, draw, align=center, scale=1},
            level 1/.style={sibling distance=4mm},
            level 2/.style={sibling distance=3mm},
            level 3/.style={sibling distance=2mm},
            level 4/.style={sibling distance=2mm},
            level distance=8mm]
            \Tree
            [.\node[inner]{$x \leq 4.5$};
                [.\node[inner]{$y \leq -4$};
                    [.\node[inner]{$x \leq -3.5$};
                        \node[leaf, top color=cluster3color!60, bottom color=cluster3color!80]{\textbf{3}};
                        \node[leaf, top color=cluster0color!60, bottom color=cluster0color!80]{\textbf{0}};    
                    ]
                    [.\node[inner]{$y \leq 4$};
                        \node[leaf, top color=cluster4color!60, bottom color=cluster4color!80]{\textbf{4}};
                        \node[leaf, top color=cluster1color!60, bottom color=cluster1color!80]{\textbf{1}};    
                    ]
                ]
                \node[leaf, top color=cluster2color!60, bottom color=cluster2color!80]{\textbf{2}};
            ]
            \node[scale=2] at (-2.5,-3.5) {$ $};
            \end{tikzpicture}
            } \label{fig:decision_tree}
     }
    \caption{
    	Example from~\cite{DFMR20}. The optimal $5$-means clustering (left) uses combinations of both features. The explainable clustering (middle) uses axis-aligned rectangles summarized by the threshold tree (right).}
    \label{fig:optimal_vs_tree}
\end{figure*}
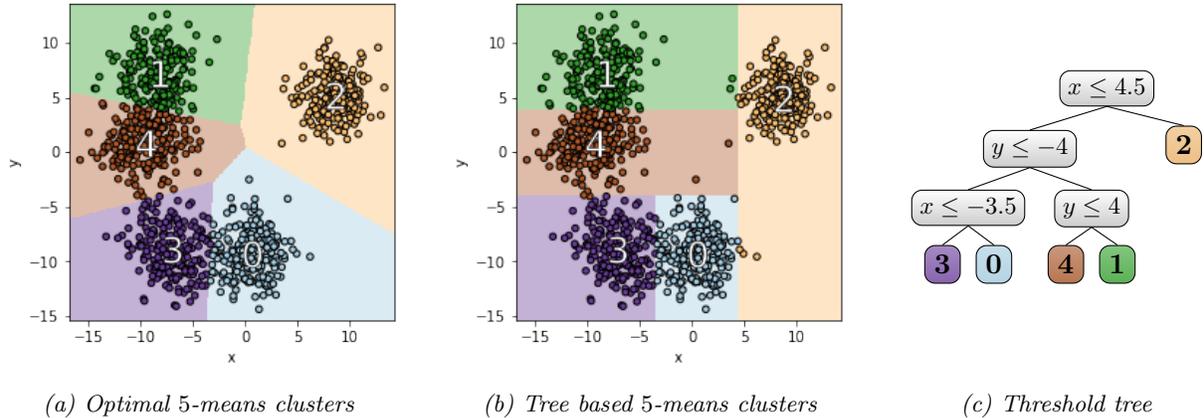

Explainability is thus a very desirable and appealing property, but the best explainable clustering may have cost much higher than the cost of the best unrestricted clusterings.  This tradeoff is captured by the \emph{price of explainability}:  the loss in cost/quality if we restrict ourselves to explainable clusterings.

In their original paper, Dasgupta et al.~\cite{DFMR20} gave a greedy algorithm which takes an arbitrary ``reference'' clustering and produces an explainable clustering from it. It repeatedly adds threshold cuts which separate the centers of the reference clustering until the threshold tree has one leaf for each center of the reference clustering. Since only the points separated in the threshold tree from their closest reference center suffer an increase in cost, %
their algorithm repeatedly selects a threshold cut that separates the
fewest points from their closest reference center.  They proved that
it outputs an explainable clustering with cost $O(k)$ times higher for the case of \kmed, and $O(k^2)$ times higher for the case of \kmeans.  They also show a lower bound of $\Omega(\ln k)$ for both of these problems.

Since the greedy algorithm's analysis is tight, %
an alternative strategy was independently proposed by~\cite{GJPS21,MakarychevS21,EMN21}:  %
take  random cuts instead!
The strategy is especially elegant in the case of \kmed (the distribution of cuts is more complex than uniform in the case of \kmeans): 
\begin{quote}
    Repeatedly select threshold cuts \emph{uniformly at random} among those that separate centers of the reference clustering. 
\end{quote}
We refer to this as the \RT algorithm (see \S\ref{sec:rand-threshold} for a formal description). 
While the algorithm is easy to describe, its performance guarantee has remained an intriguing question. 
There are simple instances in which it increases the cost by a factor of $1+H_{k-1}$, where $H_{k-1} = \nf11 + \nf12 + \nf13 + \ldots + \nf1{(k-1)}$ is the $(k-1)^{th}$ harmonic number (see \S\ref{sec:tight}), and this was conjectured to be the worst case for the \RT algorithm~\cite{GJPS21}. 

On a high level, a difficulty in analyzing the \RT algorithm is that it may take prohibitively expensive cuts with a small probability. 
To avoid this and other difficulties, the results in~\cite{GJPS21,MakarychevS21,EMN21} considered  more complex variants that intuitively forbid such expensive cuts. 
Specifically, \cite{GJPS21} gave a variant that outputs a threshold tree whose expected cost increases by at most a $O(\ln^2 k)$ factor, and both~\cite{MakarychevS21,EMN21} obtain a  better performance guarantee of $O(\ln k \ln \ln k)$ for their variants of the \RT algorithm. 
These results give an exponential improvement over that in~\cite {DasguptaFM22} but fail to settle the price of explainability, and they leave open the conjectured performance of the natural \RT algorithm.

Our main results on the price of explainability are (a)~to settle this
conjecture in the affirmative (i.e., to give a \emph{tight} analysis
of the \RT algorithm), and (b)~to show that its price of
explainability of $1+H_{k-1} = (1+o(1)) \ln k$ is not only
asymptotically correct, but also \emph{tight up to lower order terms}:
we cannot do much better regardless of the algorithm.

\begin{restatable}[Upper bound for \kmed]{theorem}{HkThm}
  \label{thm:hk-upper}
  The price of explainability for \kmed  is at
  most \( 1 + H_{k-1}\).
  Specifically, given any reference \kmed clustering, the \emph{Random Thresholds} algorithm outputs an
  explainable clustering with expected cost at most $1 + H_{k-1}$
  times the cost of the reference clustering.
\end{restatable}

\begin{restatable}[Lower Bound for \kmed]{theorem}{LBPoE}
  \label{thm:kmed-lower}
  There exist instances of \kmed for which any explainable clustering
  has cost at least $(1-o(1)) \ln k$ times the cost of the optimal
  \kmed clustering. 
\end{restatable}
These results resolve the performance of the \RT algorithm and the price of explainability  for \kmed. 

For \kmeans, we are unable to settle the price of explainability completely, but we make significant progress in closing the gap between known upper and lower bounds.  
Here, the best upper bound before our work was %
$O(k \ln k)$~\cite{EMN21} (see also~\cite{CharikarH22} for better guarantees when the input is low-dimensional).
Moreover, we know instances where any single threshold cut
increases the cost of the clustering by a factor $\Omega(k)$ (see,
e.g., \cite{GJPS21}), and hence the price of explainability of \kmeans is at least $\Omega(k)$. 

It is tempting to guess that the $O(k \ln k)$ guarantee in~\cite{EMN21} is tight, for the following reason. 
The first lower bound  $\Omega(\ln k)$  for \kmeans  in~\cite{DFMR20} is obtained by arguing that (i) a single threshold cut increases the cost by at least that of the reference clustering and (ii) a threshold tree has height $\Omega(\ln k)$, and so the total cost increases by a constant $\Omega(\ln k)$ times. 
Since we have examples where any single cut increases the cost by $\Omega(k)$, %
it is reasonable to %
hope for more complex instances to combine the two sources of hardness, and %
lose a $\Omega(k) \cdot \Omega(\ln k)$ factor.
However, we prove that this is \emph{not} the case and give an improved upper bound: %

\begin{restatable}[Upper bound for \kmeans]{theorem}{Kmeans}
  \label{thm:kmeans-main}
  The price of explainability for \kmeans is at
  most $O(k \ln \ln k)$. Specifically, given any reference
  \kmeans clustering, there
  exists an algorithm that outputs an explainable clustering with
  expected cost at most $O(k \ln \ln k)$ times the reference cost.
\end{restatable}
Hence the price of explainability for \kmeans lies between $\Omega(k)$ and $O(k \ln \ln k)$. 
We leave the tight answer as an intriguing open problem. In
particular, we conjecture that the lower bound is tight and that it is
achieved by the \kmeans variant of the \RT algorithm. (The ideas for
\kmeans also extend to 
$\ell_p$ norms for $p \geq 2$; details will appear in a forthcoming version.)

Our final contribution is to \emph{study the approximability of
  explainable clustering}.  So far, the literature has mostly focused
on settling the price of
explainability~\cite{DFMR20,LM21,GJPS21,MakarychevS21,EMN21,CharikarH22}
and its behavior in a bi-criteria setting~\cite{MakarychevS22} where
the explainable clustering is allowed to form more than $k$ clusters.
These algorithms give upper bounds on the approximability of
explainable clustering since they are all efficient, and the cost of
an optimal unconstrained clustering is a valid lower bound on the best
explainable one. Recent work of \cite{BandyapadhyayFG22,Laber22}
asked the question: \emph{how well can we approximate the best explainable clustering?} %
They showed that the problem is APX-hard, but left open the question
of whether the problem can be approximated better. Resolving this
natural question positively would have the advantage of finding good
explainable clusterings for those instances that do admit such
clusterings, which is often the experience for more practical
instances.
Our result shows a surprising hardness for the \kmed and \kmeans
problem. 
\begin{restatable}[Approximability]{theorem}{Approx}
  The explainable \kmed and \kmeans problems are hard 
  to approximate  better than $(\nf12-o(1)) \ln k$, unless P=NP. %
\label{thm:approximability}
\end{restatable}

These results show that we cannot approximate \kmed much better than
its price of explainability (unless P=NP); the approximability for
\kmeans remains tantalizingly open. %

\subsection{Outline and Technical Overview}
\label{sec:our-techniques}

\paragraph{Upper bounding the performance of the \RT algorithm.}
Our main result is the tight analysis (up to lower order terms) of the price of explainability for \kmed.
The upper bound of $1+H_{k-1} = (1+o(1)) \ln k $ is given by a tight analysis of the natural \RT algorithm.
We now sketch the main ingredients of this analysis. We start with two easy but useful observations: 
(i) by linearity of expectations, it is sufficient to bound the expected cost of a single point, 
and (ii) by translation, this point can be assumed to be the origin.
We thus reduce the problem to that of analyzing the expected distance
from the origin to the last remaining center (i.e., the center in the
same leaf of the threshold tree as the origin).
We call this process  the \emph{Closest Point Process} and define it formally in \S\ref{sec:expon-clocks-proof}.

\paragraph{Algorithm has no better guarantee than $1+H_{k-1}$.} For
this discussion, let us make a simplifying assumption that is
\emph{not without loss of generality}: the $k$ centers are located on
separate axes, so that center $i$ is at $\be_i \cdot d_i$, with  $d_1 \leq d_2 \leq \ldots \leq d_k$, hence the closest center is at a distance $d_1$.
As cuts are selected uniformly at random, the first cut removes some
center $i_1$ with probability $\nicefrac{d_{i_1}}{\sum_{j}
  d_j}$. Conditioned on that, the second cut removes center $i_2$ with
probability $\nicefrac{d_{i_2}}{\sum_{j\neq i_1}  d_j}$, and so on.
In other words, at each step, a center $i$ is separated from the
origin with probability proportional to its distance $d_i$.
For the further special case when $d_2 = d_3 = \ldots = d_k = D$, the
expected distance to the last remaining center is:
\begin{align*}
  \Pr[1\mbox{ is last center}] & \cdot d_1  + \left(1- \Pr[1 \mbox{ is last center}] \right) D \leq  d_1  + \left(1- \Pr[1 \mbox{ is last center}] \right)\cdot D  \\[2mm]
   &  \textstyle =  d_1 + \big( 1- \big(1 - \frac{d_1}{(k-1)D +1}\big) \big(1 - \frac{d_1}{(k-2)D +1}\big)\cdots  \big(1 - \frac{d_1}{D +1}\big)\big)\cdot D\,.
\end{align*}
This is an increasing function of $D$ and tends to
$(1+ H_{k-1})\cdot d_1$ when $D/d_1 \to \infty$, which shows that the
\RT algorithm cannot be better than the conjectured factor of
$1+H_{k-1}$ (see also \S\ref{sec:tight-analysis} for a formal
description).

\paragraph{Inductive argument and reduction to worst-case instances.}
But can this special setting really be the worst-case? Perhaps surprisingly, we prove that this is the case.
An inductive argument %
can help remove the assumption that $d_2 = d_3 = \ldots = d_k$:
Since $(1+H_{k-1})\cdot d_1$ is the right answer for $k \leq 2$, we
can try to proceed inductively on the number of centers to analyze
\[
  \sum_{i=1}^k \Pr[\mbox{first cut removes center $i$}] \cdot \EE[\mbox{expected cost of process with centers $[k] \setminus \{i\}$}]\,.
\]
Since each sub-instance in the sum has $k-1$ centers, we can use the
induction hypothesis to bound every term except $i = 1$ in the sum by
$(1+H_{k-2})\cdot d_1$. To bound the cost of the instance with centers
$[k]\setminus \{1\}$, we could proceed based on the following natural
observation: the farther away a center is, the smaller probability it has to
be the last remaining center, since it is more likely to be
cut/removed at each step). This would mean that the expected cost of
 the process with centers $[k] \setminus \{1\}$ is at most
$\frac{d_2 + d_3 + \ldots + d_k}{k-1}$.  And substituting, we get that
the expected distance to the last center is at most
\[
  \Pr[\mbox{first cut removes center $1$}] \cdot \frac{d_2 + d_3 + \ldots + d_k}{k-1}  + \left( 1- \Pr[\mbox{first cut removes center $1$}]\right) (1+H_{k-2}) \cdot d_1\,,  
\]
which is at most $(1+H_{k-1})\cdot d_1$ using that $\Pr[\mbox{first cut removes center $1$}] = \frac{d_1}{d_1 + d_2 + \ldots + d_k}$.

Several research groups found the above inductive proof for the
separate-axis special case, and it was one of the main motivations for
the conjectured performance of the \RT algorithm.  To prove it for the
general case, it ``only'' remains to remove the assumption that each
center is located on a separate axis.  This assumption, however, turns
out to be highly non-trivial to overcome. One indication of this difficulty is that, in the general
case, there are arbitrary correlations between centers: whether center
$i$ is removed impacts the probability that $j$ is removed.  This
causes most natural monotonicity conditions not to hold anymore.  For
example, when centers are arbitrarily located, a far center can be
more likely to be the last one than a closer one.  We overcome these
difficulties in a technical proof that manages to show that the
worst-case is as above. In this proof, we write the points as a conic
combination of cuts, view the cost as a
function of this embedding, and naturally try to bound its derivative.
This is where the technical challenges appear: since the derivative is
also not ``well-behaved'' we define a better-behaved upper bound
called the ``pseudo-derivative'', and show that this pseudo-derivative
is maximized when all points are at the same distance $D$ from the
origin (even when they are not along separate axes). We then bound the
pseudo-derivative for the non-separate-axis uniform case.  This is the
technically most challenging part of the paper, and we present it in
\S\ref{sec:tight}.

\paragraph{A Simpler proof via Competing Exponential Clocks.}
Interestingly, we can present not just one but two proofs of the
correct $(1+o(1)) \ln k$ bound: we give an alternative simpler
proof %
which takes the viewpoint of \emph{competing exponential clocks}
(previously used, e.g., for the multiway cut problem~\cite{GeHYZ11,BuchbinderNS18,SharmaV14}).  In the separate-axis case, it boils down to sampling an
exponential random variable $Z_i$ with rate $d_i$ for each center
$i$. Two well-known properties of the exponential distribution are
that (i) the probability that $i$ ``rings first'' is proportional to
its rate, i.e., $\Pr[ Z_i \leq \min_{j \neq i} Z_j] = d_i/\sum_j d_j$
and that (ii) the distribution is memoryless
$\Pr[Z_i \geq s+t \mid Z_i\geq t] = \Pr[Z_i \geq s]$.  This implies that
taking cuts in the order of the random variables $\{Z_i\}_{i\in [k]}$
(until one center remains) is identical to the Closest Point
Process. We now analyze the competing exponential clocks as follows.
For a center $i\in [k]$ with $i\geq 2$, let $Q_i$ be the probability
that $i$ is the last center among the faraway centers
$[k] \setminus \{1\}$.  Conditioning on this, $i$ is the last center
and we pay a distance of $d_i$ instead of $d_1$ if $Z_1 \leq Z_i$.
Now for the probability of $Z_1 \leq Z_i$ to be maximized, $Z_i$
should be as large as possible \emph{in the event when $i$ is last
  among $[k] \setminus \{1\}$}. So we can upper bound the contribution
of center $i$ by considering the upper quantile of the exponential
distribution of $Z_i$ with total probability mass $Q_i$. Now standard
calculations show that the total contribution of center $i$ to the
cost is $d_1\left(Q_i - Q_i \ln(Q_i)\right)$. We thus get the upper
bound
\[
    \underbrace{d_1}_{\text{contribution of close center $1$}} + \underbrace{d_1 \sum_{i=2}^k  \left(Q_i - Q_i \ln(Q_i)\right)}_{\text{contribution of far centers}} = d_1(2 + \ln(k-1))\,,
\]
where used that the entropy $\sum_{i=2}^k - Q_i \ln(Q_i)$ is at most
$\ln(k-1)$. What is particularly nice about this viewpoint is that the
analysis does not use the assumption of centers being on separate
axes.  Indeed, we can define exponential random variables for each cut
(as we did in our first proof), and the whole machinery goes through.
A small complication arises due to cuts that separate the closest
center $1$ along with other points from the origin, but we can give a
less precise but still tight (up to lower order terms) bound.  Apart
from achieving the factor $(1+o(1)) \ln k$, the arguments are also
arguably cleaner and easier than even the prior non-tight analyses of
the \RT algorithm. We present these arguments in
\S\ref{sec:expon-clocks-proof}.

\paragraph{Lower-Bounding the Price of Explainability.} 
Recall that the  $\Omega(\ln k)$ on the price of explainability for \kmed~\cite{DFMR20} is based on the following idea
\begin{enumerate}[nosep]
    \item Select $k$ centers uniformly at random from
      a hypercube $\{0,1\}^d$, and %
    \item Add a $1$-ball around each center with $d$ points, one per
      dimension, giving $dk$ points..
\end{enumerate}
The optimal unconstrained clustering has cost $dk$, so how expensive
is the best explainable clustering? Any pair of centers expect to
differ in $\nf{d}2$ coordinates, and so by concentration, their
distance $\approx d/2$ whp. Furthermore, in a sub-instance with $k'$
centers, any cut separates $k'$ points from their closest center, and
these incur cost $\approx d/2$. As the threshold tree has a height of
at least $\log_2 k$, the total cost of any explainable clustering can
now be seen to be at least $\approx (dk/2)\log_2 k$.  While
asymptotically tight, the above symmetric construction does not lead
to stronger lower bounds than $\frac12 \log_2 k$. We instead use an
asymmetric construction to achieve our tight lower bound of
$(1-o(1))\ln k$, and it gives us hardness of approximability too!
\begin{enumerate}[nosep]
    \item Place a special center at the origin, and take a $1$-ball around it giving $d$ points.
    \item The remaining centers are located at the characteristic
      vectors of some carefully chosen subsets of $\{1, \ldots,
      d\}$, and 
    \item Finally, add many points colocated with the centers which
      force any good threshold tree to have one leaf per center.
\end{enumerate}
Now the only way to separate a center from the origin is to employ a
threshold cut along a dimension, which corresponds to an element in
the set corresponding to that center. Our threshold cuts must thus
form a \emph{hitting set} of the set system corresponding to the
non-special centers. Furthermore, the number of points separated from
their closest center is equal to the size of this hitting set.  This
tight connection allows us to apply the known results for the hitting
set problem, and we get a $(1-o(1))\ln k$ lower bound on the price of
explainability for \kmed. 
In addition, the connection together with
Feige's landmark paper~\cite{Feige98} implies our hardness of
approximation results. (Interestingly enough,
\cite{BandyapadhyayFG22} give a very similar construction, but with
different parameters, which only gives them NP-hardness.) We remark that the hardness result for \kmeans
follows from that of \kmed since all points and centers are located on
the hypercube, and thus the $\ell_1$-distances equal the squared
$\ell_2$-distances.  We present the reduction from hitting set and its
implications in~\S\ref{sec:tight-lower-bounds}.

\paragraph{Improvements for \kmeans.} Our final result is an
$O(k \ln \ln k)$ price of explainability for \kmeans.  We observe
that there are two ways to achieve the weaker $O(k \ln k)$ bound. The
first transforms the \kmeans instance into \kmed, but this distorts
distances by at most $k$ using the Cauchy-Schwarz inequality; then we
lose another $O(\ln k)$ using our analyses above. Another follows the
approach of \cite{EMN21}, of finding cuts that have a good
cost-to-balance ratio. Both these approaches are tight, but we show
that they cannot be tight at the same time! I.e., if we lose a factor
of $\Omega(k)$ due to Cauchy-Schwarz, then the cuts partition the
instance into parts that are a constant factor smaller, and the loss becomes a
geometric sum that sums to $O(k)$. A quantitative version of this
tradeoff gives our result; the details appear in \S\ref{sec:kmeans}.

\paragraph{Outline.} We present the simpler exponential clocks-based
proof for \kmed in \S\ref{sec:expon-clocks-proof}, followed by the
matching hardness in \S\ref{sec:tight-lower-bounds}. The result for
\kmeans is in \S\ref{sec:kmeans}, followed by the tight $1+H_{k-1}$
bound for \kmed in \S\ref{sec:tight}.

\subsection{Further Related Work}
\label{sec:related-work}

We now discuss some of the related results beyond those mentioned
above. Some works consider the effect of the dimension $d$ of the
price of explainability. Laber and Murtinho~\cite{LM21} showed an
$O(d \ln k)$ price of explainability for \kmed, which was improved by Esfandiari et
al.~\cite{EMN21} to $O(\min\{d \ln^2 d, \ln k\ln\ln k\})$. Charikar and
Hu~\cite{CharikarH22} showed that the price of explainability is at most $k^{1-2/d} \poly\log k$ for
\kmeans, and a lower bound tight up to poly-logarithmic
terms. Esfandiari et al.~\cite{EMN21} also gave a lower bound of
$\Omega(d)$ for \kmed. Frost et al.~\cite{FrostMR20} posed the question of getting better
guarantees using more than $k$ clusters; Makarychev and
Shan~\cite{MakarychevS22} showed how to open $(1+\delta)k$ centers and
get a guarantee of $O(\nf1\delta \cdot \ln^2 k \ln\ln k)$ for
\kmeans.

The algorithmic problem has received much less attention. Bandyapadhyay et
al.~\cite{BandyapadhyayFG22} gave algorithms that find the best \kmed and \kmeans
clusterings in time $n^{2d}\cdot (dn)^{O(1)}$.  They also showed
NP-hardness, and $W[2]$-hardness of finding the best explainable
clustering; interestingly, their hardness construction is also based
on the hitting set problem and is very similar to ours, but they use a
different setting of parameters and hence only infer an NP-hardness.
Laber~\cite{Laber22} gave an APX-hardness based on a reduction from finding
vertex covers in triangle-free graphs. Our result showing a logarithmic hardness
essentially settles the
question for \kmed. 

Both the \kmed and \kmeans problems have been studied extensively in
the unconstrained setting (i.e., without the explainability
requirement), both for geometric spaces (see,
e.g.,~\cite{Cohen-AddadKM16,FriggstadRS16,Cohen-AddadEMN22,Cohen-AddadSL22}) and
general metric spaces (see, e.g.,~\cite{Cohen-Addad0LS23,ByrkaPRST17}).  The
techniques and algorithms for those settings seem orthogonal to those
used for our problems.

\subsection{Preliminaries and Notation}
\label{sec:preliminaries}

Given points $\cX=\{\bm x^1,\dots,\bm x^n\}\subseteq \RR^d$, a
\emph{clustering} $\cC$ of $\cX$ is a partition of $\cX$ into
\emph{clusters} $\{C^1,\dots,C^k\}$. %
Each cluster $C^i$ is assigned a center $\bmu^i$
(giving \emph{distinct} centers
$\cU=\{\bm \mu^1, \dots ,\bm \mu^k \}\subseteq \RR^d$). %
Let $\pi(\bm x)$ be the center $\bm \mu^j
\in \cU$ corresponding to the cluster $C^j$ containing $x$, and 
define the $q$-norm cost of a clustering $\cC$ with centers $\cU$ as
\begin{align}
  \cost_q\left(\pi,\cU \right) = \sum_{\bm x\in \cX
  }\|\bm x - \pi(\bm x) \|_q^q. \label{eq:cost}
\end{align}
The \kmed and \kmeans costs of a clustering are simply the minimum
values for the parameters $q=1$ and $2$, minimized over all possible
centers $\cU$.

\emph{Threshold Cuts and Trees.} We call a hyperplane of the form $x_i\leq
\theta$ a \emph{threshold cut}, and represent it as $(i,\theta)$. A
\emph{threshold tree} $T$ is a binary tree with each non-leaf node $u$
corresponding to a threshold cut $(i_u,\theta_u)$. Define
$B_u\subseteq \RR^d$ as the region corresponding to node $u \in T$,
where $B_r:=\RR^d$ for  $r$ being the root of $T$; if nodes $l(u)$ and $r(u)$ are the left and right children of node $u$, then
\begin{align*}
B_{l(u)}:=B_{u}\cap \{\bm x \mid x_{i_u}\leq \theta_u \} \qquad
           \text{and} \qquad
B_{r(u)}:=B_{u}\cap \{\bm x \mid x_{i_u}> \theta_u \}. 
\end{align*}

\emph{Explainable Clusterings.}  Given points $\cX$ and a threshold
tree $T$, the clustering $\cC_T$ of $\cX$ \emph{explainable} by the
threshold tree $T$ is the partition of $\cX$ induced by the regions
corresponding to leaves in $T$, i.e., each leaf $\ell$ of $T$
generates a cluster $C^\ell := \cX\cap B_\ell$ of $\cC_T$. A
clustering $\cC$ of $\cX$ is said to be an \emph{explainable
  clustering} if there exists a threshold tree $T$ such that
$\cC=\cC_T$.

For a set of centers $\cU$, a threshold tree $T$ \emph{separates}
$\cU$ if each of the regions corresponding to leaves in $T$ contains
exactly one center in $\cU$. Let $\mu^\ell$ denote the unique center
in the singleton set $\cU \cap \cal B_\ell$ for leaf $\ell$ in $T$. For any
set of points $\cX$, centers $\cU$, and a threshold tree $T$ that
separates $\cU$, each leaf in $T$ corresponds to a cluster $C^{\ell}$
in the clustering $\cC_T$, and also to a center $\mu^\ell$. Such a
tree induces an assignment $\pi_T: \cX \to \cU$ from points to
centers. With this, we can define
\begin{align}
  \cost_q\left(T \right) = \cost_q(\pi_T, \cU) = \sum_{\bm x\in \cX }\|\bm x - \pi_T(\bm x) \|_q^q.
\end{align}

\section{Explainable \kmed via Exponential Clocks}
\label{sec:expon-clocks-proof}

We now give a bound of $(1+o(1))\, \ln k$ on the price of
explainability for \kmed. This is slightly weaker than the bound of
$1 + H_{k-1} \approx \ln k + O(1)$ promised in \Cref{thm:hk-upper},
but the proof is simpler and more illuminating. (We give the proof of
the tight bound in \S\ref{sec:tight}.)

\subsection{The Random Threshold Algorithm and the Closest Point Process}
\label{sec:rand-threshold}

Let us first formalize the \RT algorithm: given a reference clustering
for point set $\cX$ which opens centers $\cU$ and maps the data points
to centers using $\pi: \cX \to \cU$, we construct a threshold tree $T$
randomly as follows. For simplicity, let $\cX \subseteq  [a,b]^d$ for some
$a,b\in \RR$. We start with the trivial threshold tree with the root
corresponding to all of $\RR^d$. Now while the leaves of $T$ do not give
us a separating partition for $\cU$, we pick a dimension $i \in [d]$
and a value $\theta \in [a,b]$ independently and uniformly at
random. For each leaf $u$ of $T$, if this threshold cut separates at
least one pair of centers which share the region $B_u$, partition the
leaf using the threshold cut. It is easy to see that as long as all
the centers in $\cU$ are distinct, this process outputs a threshold
tree that separates $\cU$. The main question is: what is the cost of
the resulting explainable clustering $\cC_T$, in expectation?

Since the algorithm does not depend on the data points $\cX$, and it
is invariant under translations and scaling, we can use linearity of
expectations and focus on the following simpler problem:

\begin{defn}[Closest Point Process]
  Given a set of $k$ points $\Univ \sse \RR^d$, let
  $\bm p^* := \arg\min_{\bm p \in \Univ} \| \bm p \|_1$ be the point in $\Univ$ closest to
  the origin. Assume $\|\bm p^*\|_1 = 1$. Run the \RT algorithm to create
  a random threshold tree $T$ that separates this point set $\Univ$.
  Consider leaf node $u \in T$ whose corresponding region
  $B_u \sse \RR^d$ contains the origin, and let $\widehat{\bm p}$ be the
  unique point of $\Univ$ in this region $B_u$.  Define
  \begin{gather}
    f(\Univ) := \EE [\; \| \widehat{\bm p} \|_1 \;]. \label{eq:eff}
  \end{gather}
  Finally, define $\alpha(k) := \max_{\Univ: |\Univ|=k} f(\Univ)$. 
\end{defn}

\begin{lemma}[Focus on Closest Point]%
\label{lem:single-point}
  Given a reference clustering $\pi: \cX \to \cU$, the expected cost
  of the explainable clustering produced by the \RT algorithm is
  \[ \EE[ \cost(\pi_T, \cU) ] \leq \alpha(|\cU|) \cdot \cost(\pi,\cU)\,. \]
  Therefore, the price of explainability is at most $\alpha(k)$.
\end{lemma}
Given this reduction (which we prove in~\Cref{app:clocks}), the main result of this section is:
\begin{theorem}[Exponential Clocks]
  \label{thm:clocks}
  For any set $\Univ$ with $k$ points, $f(\Univ) \leq (1+o(1))\, \ln k$. 
\end{theorem}

\subsection{The Exponential Clocks Viewpoint: the Last Point}
\label{sec:exp_clocks}
We now focus on bounding the value $f(\Univ)$ for any point set
$\Univ \in \RR^d$. We first impose some structure, just for the sake of
analysis. Since $\ell_1$ metrics can be written as a non-negative sum
of cut metrics (see, e.g.,~\cite{DL97}), again using the
data-obliviousness and translation-invariance of the algorithm we can
assume the following without loss of generality (see
\Cref{sec:cut-metrics}). %
\begin{enumerate}[nolistsep]
\item there are $d = 2^k$ dimensions (one for each subset $S \sse \Univ$
  of the points), and
\item the instance is specified by non-negative values
  $\{z_S\}_{S \sse \Univ}$ such that for each point $\bm p \in \Univ$, it lies at
  location \[ \bm p_S := z_S \mathbf{1}(\bm p \in S). \]
  Hence the distance of a point $\bm p$ is $\|\bm p\|_1 = \sum_S z_S
  \mathbf{1}(\bm p \in S) = \sum_{S: \bm p \in S} z_S$.
\end{enumerate}

Given this structure and the focus on $f(\Univ)$, we need to analyze the
following process:
\begin{quote}
  \emph{The Last Point Process.} Start with some set $\Set \sse \Univ$, and
  empty sequence $\cS \gets \langle\rangle$. At each step, pick a set
  $S \not\in \cS$ with probability
  $\frac{z_S}{\sum_{T \not\in \cS} z_T}$ and add it to the end of
  $\cS$. If $|\Set \setminus S| \neq \emptyset$, set
  $\Set \gets \Set \setminus S$. When all remaining sets $S \not\in \cS$
  have $z_S = 0$, stop and output the current $\Set$, a singleton set we
  call $\Set_{\text{final}}$.
\end{quote}
An inductive argument shows that if we start with $\Set = \Univ$, the final
set $\Set_{\text{final}}$ has the same distribution as the set of points
in the region containing the origin in the \RT
algorithm. Specifically, the first cut is taken with probability $\frac{z_S}{\sum_{T \not\in \cS} z_T}$ and the process inductively proceeds; the  process is thus  identical to the Closest Point Process, and so   $\Set_{\text{final}}$ contains a single point
$\widehat{\bm p} \in \Univ$ when the process stops with
$f(\Univ) = \EE[ \|\widehat{\bm p}\|_1]$.

To analyze this, we change the perspective slightly further, and
recast the process in terms of ``exponential clocks''. Define
independent exponential random variables $X_{S}\sim \exp(z_{S})$ for
each set $S \sse 2^\Univ$ such that $z_{S} > 0$. Since exponential random
variables $\{Y_i \sim \exp(r_i)\}$ have the memorylessness property, and the property that
$\Pr[ Y_j = \min_i \{Y_i\} ] = \frac{r_j}{\sum_i r_i}$, we see the
sets in the same order $\cS$ as in the last-center process
above. Moreover, this
order depends only on the set $\Univ$, and is independent of the starting
set $\Set \sse \Univ$.

Now consider the Last Point Process starting with different sets $\Set
\sse \Univ$
(and not just the entire point set $\Univ$): naturally, the identity of
the final point $\widehat{\bm p}$ changes. However, it turns out we can
make the following claim. Define the event \emph{point $\bm p \in \Univ$ is
  last in $\Set$} if starting with the set $\Set$ results in
$\Set_{\text{final}} = \{\bm p\}$. It turns out that being last in this
process has a nice ``monotone'' property. (We defer the proof to~\Cref{app:last-last}.)

\begin{restatable}[Monotonicity]{lemma}{LastPointLemma}
  \label{cor:last-last}
  For any  sets $T,\Set$ such that $T\subseteq \Set$, and any point $\bm p\in
  V\setminus T$, we have
  \begin{align*}
    \textnormal{``$\bm p$ is last in $\Set$''}\Rightarrow \textnormal{``$\bm p$ is last in $\Set\backslash T$''}.
  \end{align*}
\end{restatable}

\subsection{Bounding the Expected Cost}
\label{sec:clocks-cost}

By the definition of our process, we know that
\begin{align}
	f(\Univ) &=\sum_{\bm p\in \Univ}\|\bm p\|_1\cdot \Pr[\textnormal{$\bm p$ is last in $\Univ$}]
	\leq \gamma +\sum_{\bm p: \|\bm p\|_1>\gamma}\|\bm p\|_1\cdot
   \Pr[\textnormal{$\bm p$ is last in $\Univ$}] \label{eq:1}
\end{align}
for any $\gamma$. (We choose $\gamma >1$, which ensures that $\bm p\neq \bm p^*$.)
We now bound~(\ref{eq:1}) as follows.
Observe that whenever $\bm p$ is last in $\Univ$, the following is true. There must exist a cut $T$ that removes the closest point $\bm p^*$ before any cut removes  $\bm p$, i.e.,  $\bm p^* \in T, \bm p \notin T$. This implies  $X_T\leq X_S$ for all
  sets $S$ such that $\bm p\in S, \bm p^*\notin S $, which can be written as
  \[
  X_T \leq X_{\bm p}\mbox{, where  }X_{\bm p}:= \min_{S: \bm p\in S,\bm p^*\notin S} X_S\,.
  \]
  Second, by the Monotonicity \Cref{cor:last-last}, we have that $\bm p$ is last in $\Univ$ implies that $\bm p$ is last in $\Univ \backslash T$.
Defining $\cF_{\bm p} := \{ T
\mid \bm p^* \in T, \bm p \not \in T\}$ to be all those cuts that could remove $\bm p^*$ before $\bm p$,  therefore yields the upper bound
\begin{align}
	f(\Univ) &\leq \gamma + \sum_{\bm p: \|\bm p\|_1>\gamma}
                          \|\bm p\|_1\cdot \Pr[\exists T \in \cF_{\bm p}
                          \textnormal{ such that } X_{T}\leq X_{\bm p}
                          \bigwedge \bm p \textnormal{ is last in } \Univ\backslash T] \notag \\
	&\leq  \gamma + \sum_{\bm p: \|\bm p\|_1>\gamma }\|\bm p\|_1\sum_{T \in \cF_{\bm p}}\Pr[ X_{T}\leq X_{\bm p} \bigwedge  \bm p \textnormal{ is last in } \Univ\backslash T]\,. \tag{union bound}
\end{align}
We  upper bound the contribution of a fixed point $\bm p$ to the above expression. By the law of total probability, $\Pr[ X_{T}\leq X_{\bm p} \bigwedge  \bm p \textnormal{ is last in } \Univ\backslash T]$ equals
\begin{align}
     \int_{-\infty}^{\infty} \Pr[X_T\leq t \bigwedge \bm p \textnormal{ is last in }    \Univ\backslash T \mid X_p = t] f_{X_p} (t) dt\,,
\end{align}
where $f_{X_p} (t)$ denotes the probability density function of $X_p$.
The event $X_T\leq t$ is independent from the event ``$\bm p$ is last in
$\Univ\backslash T$'' because, $T$ does not cut any points in $\Univ\backslash
T$ and hence the value of $X_T$ is irrelevant to the process
restricted to points in $\Univ\backslash T$. We also know that $X_T$ and
$X_{\bm p}$ are independent. These observations can be used to rewrite the
above expression as
\begin{align}
     \int_{-\infty}^{\infty} \Pr[X_T\leq t]\cdot \Pr[ \bm p \textnormal{ is last in }    \Univ\backslash T \mid X_p = t] f_{X_p} (t) dt\,.
     \label{eq:beforecdfandpdf}
\end{align}
As $\Pr[X_T \leq t]$ is an increasing function of $t$, the above expression is maximized if the probability mass of the event  ``$\bm p$ is last in $\Univ \backslash T$'' is on large values of $t$. Formally, if we select $\btau$ to be threshold so that 
\begin{align}
    \int_{-\infty}^{\infty} \Pr[ \bm p \textnormal{ is last in }    \Univ\backslash T \mid X_p = t] f_{X_p} (t) dt = \int_{-\infty}^{\infty} \mathbf{1}[t\geq \btau] f_{X_p} (t)dt  = \int_{\btau}^{\infty}  f_{X_p} (t)dt 
    \label{eq:tau-selection}
\end{align}
then \eqref{eq:beforecdfandpdf} is upper bounded by 
\begin{gather}
  \int_{\btau}^{\infty} \Pr[X_T\leq t] f_{X_p} (t) dt \,. \label{eq:ubd}
\end{gather}
To understand this expression, recall that $X_T$ is an exponential random variable with rate $z_T$. Further, the random variable $X_{p}$ is the minimum of exponentials, and
hence is itself exponentially distributed with rate %
${\ell}(\bm p) = \sum_{S: \bm p \in S, \bm p^* \not \in S} z_S$. In other words, $\Pr[X_T \leq t] = 1- e^{- z_T \cdot t}$ and $f_{X_p}(t) = {\ell}(\bm p) e^{-{\ell}(\bm p) t}$ for $t\geq 0$.  This gives us that the choice of $\btau$ that satisfies the identity~\eqref{eq:tau-selection} is 
\[
    \btau = \frac{-\ln Q_T(\bm p)}{{\ell}(\bm p)}, \mbox{ where  $Q_T(\bm p)$ is the probability that $\bm p$ is last in $\Univ\backslash T$.}
\]
The integral~(\ref{eq:ubd}) can be upper-bounded by standard calculations:
\begin{align*}
    \int_{\btau}^{\infty} \Pr[X_T\leq t] f_{X_p} (t) dt & = 	 \int_{\btau}^\infty  (1-e^{-z_{T}t})\cdot {\ell}(\bm p)\cdot e^{-{\ell}(\bm p)t} dt \\
	&= Q_T(\bm p)- \frac{{\ell}(\bm p)}{{\ell}(\bm p)+z_T}e^{-({\ell}(\bm p)+z_T)\cdot \btau}\\
	& \leq Q_{T}(\bm p)\left (1-\frac{{\ell}(\bm p)}{{\ell}(\bm p)+z_T}\cdot \left(1+\frac{z_T \ln Q_T(\bm p)}{{\ell}(\bm p)} \right)\right )\\
	&= Q_T(\bm p)\cdot z_T\left ( \frac{1-\ln Q_T(\bm p)}{{\ell}(\bm p)+z_T} \right ).
\end{align*}
Substituting in this upper bound, we have 
\begin{align*}
  f(\Univ)
  &\leq \gamma + \sum_{\bm p: \|\bm p\|_1>\gamma
    }\|\bm p\|_1\sum_{T \in \cF_{\bm p}}Q_T(\bm p)\cdot z_T\left ( \frac{1-\ln Q_T(\bm p)}{{\ell}(\bm p)+z_T} \right ) \\
  & \leq \gamma + \sum_{\bm p: \|\bm p\|_1>\gamma
    }\frac{\|\bm p\|_1}{\|\bm p\|_1-1}\sum_{T \in \cF_{\bm p}}Q_T(\bm p)\cdot z_T\left
    ({1-\ln Q_T(\bm p)} \right ), 
  \intertext{where we use that ${\ell}(\bm p) + z_T \geq
  \|\bm p\|_1-1$. Indeed ${\ell}(\bm p) =\sum_{S: \bm p \in S, \bm p^* \not \in S} z_S \geq \|\bm p\|_1 - \|\bm p^* \|_1 \geq \|\bm p\|_1 - 1$. Using the fact that $\nicefrac{x}{x-1}$ is a decreasing function and then replacing the summation over all $\bm p:\|\bm p\|_1> \gamma $ to all $\bm p\neq \bm p^*$, and exchanging the summations gives } 
  & \leq \gamma + \frac{\gamma}{\gamma-1}\sum_{T\ni \bm p^*} z_T
    \sum_{\bm p\in \Univ\backslash T}Q_T(\bm p) (1 - \ln Q_T(\bm p)).
\end{align*}
Observe that for any cut $T$, we have
$\sum_{\bm p \in \Univ \setminus T} Q_T(\bm p) = 1$, and the sum is over at most  $k-1$
points. As the entropy $\sum_{\bm p\in \Univ\backslash T}Q_T(\bm p) (- \ln Q_T(\bm p))$ of $|\Univ\backslash T|\leq k-1$ outcomes is at most $\ln(k-1)$, the inner sum is at most $1+
\ln(k-1)$. Finally, using that $\sum_{T \ni \bm p^*} z_T = \|\bm p^*\|_1 = 1$,
we get
\begin{align*}
	f(\Univ)  &\leq \gamma +\frac{\gamma}{\gamma
                                -1}(1+\ln(k-1)) 
                              \leq \ln(k-1)+2\sqrt{1+\ln(k-1)}+2
\end{align*}
by optimizing over $\gamma$. This proves~\Cref{thm:clocks}, and gives
us an asymptotically optimal bound on the price of explainability. In
the next section we show that the bound is, in fact, tight up to
lower-order terms.

\section{Lower Bounds on the Price of Explainability}
\label{sec:tight-lower-bounds}

In this section, we prove  a tight lower bound on the price of explainability (up to lower order terms), and a lower bound on the approximability of explainable clustering. Both results are obtained via a reduction from the classic \emph{hitting set problem}: given a set system $([d], \cS)$, where $[d] = \{1, \ldots, d\}$ denotes the ground set and $\cS = \{S_1, S_2, \ldots, S_k\}$ is a family of $k$ subsets of $[d]$, the task is to find the smallest subset $H \subseteq [d]$ that hits every subset in $\cS$, i.e., $H\cap S_i \neq \emptyset$  for all $S_i \in \cS$. We further say that a hitting set instance $([d], \cS)$ is \emph{$s$-uniform if all subsets of $\cS$ are of the same size $s$}. We now first present the reduction from $s$-uniform hitting set instances to explainable clustering, and we then analyze its implications.

\paragraph{Reducing hitting set to explainable clustering.}
Given an $s$-uniform hitting set instance $([d], \cS = \{S_1, S_2, \ldots, S_k\})$,  define the following data set $\cX$
  in $\{0,1\}^d$ and reference solution $\cU$:
  \begin{enumerate}
  \item The reference clustering has $k+1$ centers $\cU := \{\bmu_0, \bmu_1, \ldots, \bmu_k\}$, where $\bmu_0$ is at the origin, and each other
    $\bmu_i \in \{0,1\}^d$ is the characteristic vector of the set
    $S_i$.
  \item The data set $\cX$ consists of one point at each of the
    locations $\{\mathbf{e}_i\}_{i \in [d]}$, and
    $M = \poly(d,k) \gg \max\{d,k\}$ ``colocated'' points at each of the $k+1$
    locations in $\cU$, giving $|\cX| = d + M\cdot (k+1)$.
  \end{enumerate} 
The cost of this reference clustering $\cU$ with $k+1$ centers is at most $d$, since all the $d$ non-colocated points can be assigned to the center $\bmu_0$.
We proceed to analyze the cost of an optimal \emph{explainable} clustering with $k+1$ centers.
\begin{lemma}
    Let $h$ be the size of an optimal solution to the hitting set instance $([d], \cS)$ and let $\opt$ be the cost of an optimal $(k+1)$-median explainable clustering of the data set $\cX$. Then
    \[
       d +  h(s-2-o(1))  \leq \opt \leq d  + h(s-2)\,.
    \]
    Moreover, the same bounds hold for the optimal $(k+1)$-means explainable clustering.
    \label{lemma:hitting-set-reduction}
\end{lemma}
\begin{proof}
    We present the proof for $(k+1)$-median and then observe that all the distance calculations also hold for $(k+1)$-means since all the points of $\cX$ have binary coordinates and, as we will show, the centers will be (arbitrarily close) to such coordinates as well.  Note that $\|\mathbf{p}-\mathbf{q}\|_1 = \|\mathbf{p}-\mathbf{q}\|_2^2$ if $\mathbf{p},\mathbf{q}\in \{0,1\}^d$. We now proceed with the analysis for $(k+1)$-median. 

    The $M$ points colocated with each of the reference centers $\bmu_i$
  ensure that the best explainable clustering separates each of the
  centers $\bmu_i$. Separating a center $\bmu_i$ from $\bmu_0$ using
  a threshold cut means choosing some dimension $j \in S_i$ and a value
  $\theta \in (0,1)$, which in turn also separates the data point $\be_j$
  from $\bmu_0$. Since $M \gg k$, the center for the final
  cluster containing $\be_j$ is located at some location very close
  the reference center in it, and hence this data point now incurs
  cost $s-(1+o(1)))$  instead of $1$. Here we used the fact that  each set has size $s$ and the term $o(1)$ accounts for the potential small  difference in the locations of the centers in the explainable clustering compared to those in the reference clustering.  
  The above observations imply that 
  \begin{itemize}
    \item  the collection of threshold cuts that separate $\bmu_0$ from other centers must form a hitting set for the set system $([k], \cS)$; and  
    \item if this hitting set has size
  $h'$, the cost of the explainable clustering is at least
  $h'(s-(1+ o(1))+(d-h') = d+h'(s-2-o(1))$.
    \end{itemize}
    We thus have $\opt \geq d+h(s-2-o(1))$ since $h$ is the smallest size of a hitting set. 

    For the upper bound $\opt \leq d+ h(s-2)$, let $H =\{i_1, i_2,\ldots, i_h\} \subseteq [d]$ be an optimal hitting set of size $h$. 
    Starting with the reference clustering $\cU$, build a threshold tree by  adding the threshold cuts  along dimensions $i_1, i_2, \ldots, i_h$ with thresholds $1/2$. 
    Specifically, the cut along dimension $i_1$ is at the root of the tree and the remaining cuts are recursively added to the subinstance that contain the reference center $\bmu_0$. 
    After adding these cuts we have separated $\bmu_0$ from all other centers, since $H$ is a hitting set. 
    Furthermore, the only points in $\cX$ that are separated from their closest center in $\cU$ are $\be_{i_1}, \be_{i_2}, \ldots, \be_{i_h}$. 
    Note that the tree may still contain centers $\bmu_i,\bmu_j$ with $i,j \geq 1$ that are yet not separated. 
    But they can be separated without incurring any additional cost since, in their subinstance, all points of $\cX$  are colocated with the centers (or have already been separated from their closest center $\bmu_0$). 
    Hence, we can build a threshold tree that has one leaf per center in $\cU$ and the only points of $\cX$ that are separated from their closest center are $\be_{i_1}, \be_{i_2}, \ldots, \be_{i_h}$. Each of these separated points $\be_{i_j}$ has cost at most $s-1$ instead of $1$ since the hitting set instance was $s$-uniform and the final center $\bmu_q$ that $\be_{i_j}$ is assigned to correspond to a set $S_q$ that contains $i_j$, and so $\|\be_{i_j} - \bmu_q\|_1 = s-1$. 
    Hence, the total cost of the clustering is at most $h \cdot (s-1) + d - h = d + h(s-2)$, which completes the proof of the lemma.
\end{proof}

Having described our reduction, we now proceed to its implications.

\paragraph{Price of explainability for $k$-median.} 
As aforementioned, the cost of the reference clustering $\cU$ is at most $d$. Furthermore, \Cref{lemma:hitting-set-reduction} says that the optimal $(k+1)$-median explainable clustering costs at least $h(s-2-o(1)) + d$, where $h$ is the size of an optimal hitting set of $([d], \cS)$. 
   It thus suffices to construct a set system $\cS$ having large
  $h(s-2-o(1)) \approx hs$. For example, letting $d = |\cS| = k$ and defining $\cS$ based on the Hadamard
  code would give us $s = k/2$ and a hitting set of size $\log_2 k$,
  and hence a lower bound of $\approx \frac12 \log_2 k$. A better
  guarantee follows using a probabilistic construction (selecting uniformly at random sets), whose proof  we
  defer to the appendix.
  \begin{restatable}[Hitting Set Lemma]{lemma}{SetSystem}
    \label{lem:random-construction}
    For large enough $k$, there exist  set systems $([k], \cS)$ with
    $k$ sets of size $s$ each, such that the minimum hitting set
    satisfies $h(s-2-o(1))/k \geq \ln k - O(\ln \ln k)$.
  \end{restatable}
    Combining the above lemma with our reduction shows that  the price of explainability is at least  $(1-o(1)) \ln k$,
  giving the proof of \Cref{thm:kmed-lower}.
  \LBPoE*

\paragraph{Hardness of approximation}

    Our reduction from the hitting set problem to explainable clustering immediately leads to a hardness result as well. Feige, in his landmark paper~\cite{Feige98}, proved that it is hard to distinguish whether an $s$-uniform hitting set instance $([d], \cS)$~\footnote{We remark that the result in~\cite{Feige98} is stated in the terminology of set cover. The instances constructed there has a ground set of size $n$ and a family of $m$ subsets. Furthermore, they can be assumed to be regular: each element is contained in $s$ subsets and each subset is of size $\ell$. Now in the yes case, there is a set cover so that each element is covered by exactly one set. By the regularity, this implies that the set cover has size $n/\ell = m/s$ in the yes case. Here we used that  $n\cdot s= m\cdot \ell$. In the no case however, any set cover is  at least a factor $(1-o(1))\ln n$ larger. Now in the terminology of hitting set, this is an hitting set instance with $d=m$ elements and a family $\cS$ of $k= n$ many sets, each of size $s$,  with the stated yes case and no case. }  
\begin{itemize}[nosep]
  \item \emph{(yes case:)} has a hitting set of $d/s$ elements; or
  \item \emph{(no case:)} any hitting set has size at least $(1-o(1))\ln(k)\cdot d/s$, where $k = |\cS|$.
\end{itemize}

Here, ``hard'' means that there is no polynomial-time algorithm can distinguish between these two cases unless $P=NP$; it was under the stronger assumption in Feige's original paper~\cite{Feige98} and then subsequently improved to hold under the weakest possible assumption $P \neq NP$~\cite{DinurS14,Moshkovitz15}.

Our reduction runs in polynomial time so the above hardness together with \Cref{lemma:hitting-set-reduction} implies the following.  Assuming $P\neq NP$, there is no polynomial-time algorithm that, given a data set $\cX \subseteq \mathbb{R}^d$, distinguishes whether
\begin{itemize}[nosep]
  \item \emph{(yes case:)} there is an explainable clustering with $k+1$ clusters of cost at most $2d$; or 
  \item \emph{(no case:)} any such clustering has cost at least $(1-o(1))\ln(k) d$.
\end{itemize}
As any approximation algorithm with better guarantee than $(\nf12-o(1))\ln(k)$ would allow us to distinguish between the two cases, we have the following hardness of approximation result for explainable clustering.
\Approx*

The above hardness result settles the approximability of explainable \kmed  up to small constants: it is the same as its price of explainability! For \kmeans, the situation is different. Our hardness of approximation result is far from the lower bound $\Omega(k$) on its price of explainability. 
We conjecture that there is no such hardness result matching $\Omega(k)$ and, in contrast to \kmed, that there are significantly better approximation algorithms for explainable \kmeans than its price of explainability.

\section{Explainable $k$-means clustering}
\label{sec:kmeans}

We now prove our improved bound on the price of explainability of the \kmeans
problem, which improves on the previous bound of $O(k \ln k)$. Our
main result is the following: 

\begin{theorem}
  \label{thm:k-means_poe}
  Given a data set $\cX$ and a base clustering with centers $\cU$ and
  map $\pi$, we can output a random threshold tree $T$ separating
  $\cU$ such that
  \begin{align*}
    \EE[\cost_2(T)]\leq O(k\ln\ln k)\cdot \cost_2(\pi,\cU).
  \end{align*}
\end{theorem}

At a high level, the approach is similar to that for \kmed: we give an
algorithm to separate a given set of centers, but since we are dealing
with squared Euclidean distances, we choose cuts from a non-uniform
distribution over dimensions and coordinate values. However, since a
single cut can increase the cost by a factor of $\Omega(k)$ we have to
be careful not to lose another factor of $\Omega(\ln k)$ due to the
recursion. Here we use a win-win analysis: we define a quantity called
the \emph{stretch} of a pair of points and argue that the loss due to
a single cut is just the stretch: moreover, we show that if stretch is
large, the recursive problems are relatively balanced and the loss in
the recursion is a geometric sum, adding up to $\approx O(k)$. On the
other hand, if the stretch is low, we lose less-than-the-worst-case in
each round (although we now need to take a collection of ``bulk''
cuts).

\subsection{The Closest Point Process Again}

Recall that we proved the performance of the \RT algorithm for \kmed
by reducing to the perspective of a single data point and analyzing
the expected increase in its cost. We can also define the closest
point process for any $\ell_q$ norm and for any algorithm $\cA$
separating point sets $\Univ$ that is invariant under translations and
scaling, as follows:

\begin{defn}[$\ell_q$-Norm Closest Point Process]
  Given a set of $k$ points $\Univ \sse \RR^d$, let
  ${\bm p}^* := \arg\min_{\bm{p} \in \Univ} \| \bm p \|_q$ be the point in $\Univ$ closest to
  the origin according to the $\ell_q$ metric. Assume $\|\bm p^*\|_q =
  1$. Run the algorithm $\cA$ to create a random threshold tree $T$
  that separates this point set $\Univ$. Consider leaf node $u \in T$
  whose corresponding region $B_u \sse \RR^d$ contains the origin, and
  let $\widehat{\bm p}$ be the unique point of $\Univ$ in this region $B_u$.
  Define 
  \begin{gather}
    f_{q,\cA}(\Univ) := \EE [\; \| \widehat{\bm p} \|_q^q \;]. \label{eq:eff-two}
  \end{gather}
  Finally, define $\alpha_{q,\cA}(k) := \max\limits_{\Univ: |\Univ|=k} f_{q,\cA}(\Univ)$. 
\end{defn}

A proof identical to that of \Cref{lem:single-point} shows that the
price of explainability for $\ell_q$-norm clustering is at most
$\alpha_{q,\cA}(k)$. In the rest of this section, we give some
terminology and then an algorithm $\cA$ which separates the input
point set $\Univ \sse \RR^d$; we then bound the resulting value of
$f_{q,\cA}(\Univ)$.

\subsection{Terminology}
\label{sec:terminology-kmeans}

\begin{wrapfigure}{R}{0.5\textwidth}
\centering
\begin{tikzpicture}[scale=1.5]
    \draw [<->, thick] (0,2) node (yaxis) [above] {2} |- (3,0) node (xaxis) [right] {1};
  
    \coordinate (x) at (0.8, 0);
    \coordinate (y) at (2.3, 0);
    \coordinate (a) at (0.5, 1.2);
    \coordinate (b) at (1.3, 1.8);
    \coordinate (c) at (1.8, 0.6);
    \coordinate (d) at (2.8, 0.8);
    
    \draw[dashed] (xaxis -| x) node[below] {$x$};
    \draw[dashed] (xaxis -| y) node[below] {$y$};
    
    \draw[dashed] (xaxis -| b) node[below] {$v^4_1$};
    \draw[dashed] (b) -- (1.3,0);
    
    \draw[dashed] (xaxis -| c) node[below] {$v^2_1$};
    \draw[dashed] (c) -- (1.8,0);
    
    \draw[dashed] (xaxis -| a) node[below] {$v^1_1$};
    \draw[dashed] (a) -- (0.5,0);
    
    \draw[dashed] (xaxis -| d) node[below] {$v^3_1$};
    \draw[dashed] (d) -- (2.8,0);
    
    \fill[black] (a) circle (1pt) ;     
    \fill[black] (b) circle (1pt) ;     
    \fill[black] (c) circle (1pt) ;     
    \fill[black] (d) circle (1pt) ;     
    \draw [fill=white] (x) circle (1pt) ;     
    \draw [fill=white] (y) circle (1pt) ;     
     
    \draw (a) node [above]{\(v^1\)};
    \draw (b) node [above]{\(v^4\)};
    \draw (c) node [above]{\(v^2\)};
    \draw (d) node [above]{\(v^3\)};
    \draw (x) node [above]{};
    \draw (y) node [above]{};
\end{tikzpicture}
    \caption{Intervals defined by projections.}
    \label{fig:intervals-on-dim}
\end{wrapfigure}
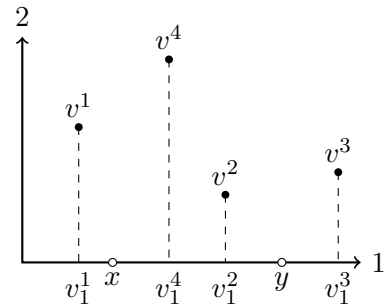

We use similar terminology as in~\cite{GJPS21}.
Given a set $\Univ \sse \RR^d$ of points, and a dimension $i \in [d]$,
let $\ell_i := \min_{v \in \Univ} v_i$ and
$u_i:= \max_{v \in \Univ} v_i$ be the leftmost and rightmost
coordinates of points. Given two values $x,y \in \RR$, let
\(\intervals_i(x, y) \) be the set of consecutive intervals along the
\(i\)-th dimension delimited by the coordinates \(x\) and \(y\)
themselves and the projections of points in $\Univ$ that lie between
\(x\) and~\(y\). For example, consider the \(2\)-dimensional point set
$\Univ$ shown in \Cref{fig:intervals-on-dim} (the same example was given in~\cite{GJPS21}).  On the horizontal axis,
two coordinate values \(x\) and \(y\) are marked along with the
projections of the points: \(\intervals_1(x,y) \) consists of the
three consecutive intervals \( [x, v^4_1], [v^4_1, v^2_1]\),
and~\([v^2_1, y]\). 

By the definition of \(\intervals_i(x, y ) \), we have
\(|x - y| = \sum_{[a,b] \in \intervals_i(x, y )} |b - a| \).  Let
\[ \intervals_{\all} := \cup_{i\in [d]} \left\{ (i, [a,b]) \mid [a,b]
    \in \intervals_i(\ell_i, u_i) \right\} \] denote the collection of
all dimension-interval pairs which are delimited by the projections of
the points onto the respective dimensions; for brevity, define $\intervals_i :=
\intervals_i(\ell_i, u_i)$. Define
\[L_2 := \sum_{(i,[a,b]) \in \intervals_{\all}} |b-a|^2. \]

A key definition is that of the \emph{pseudo-distance}: for points
\(\bm x, \bm y \in \RR^d \), let \[\intervals(\bm x, \bm y ) = \bigcup_{i\in[d]}
  \{ (i, [a,b]) \mid [a,b] \in \intervals_i(x_i,y_i ) \}. \] 
We then define the pseudo-distance between \(\bm x\) and \(\bm y\) as 
\[d_2(\bm x, \bm y ) = \sum_{(i, [a,b]) \in \intervals(x, y )} |b-a|^2.\]
It follows that \(\|\bm x-\bm y\|^2_2 \geq d_2(\bm x, \bm y ) \geq
\frac{1}{|\Univ \cup \{\bm x, \bm y\}|-1} \cdot \|\bm x- \bm y\|_2^2 \). %

We define a distribution $D_2$ 
as follows: first 
select a dimension--interval pair \((i,[a,b]) \in \intervals_{\all}\)
with probability  \( |b-a|^2/L_2 \), and then pick \(\theta \in [a, b] \) randomly such that the p.d.f.\  is
\begin{gather}
  P_{a,b}(\theta) := \frac{4}{(b-a)^2} \min(\theta - a, b - \theta)
  \quad = \quad
  \frac{4}{(b-a)^2} \min_{\bm v \in \Univ} \{|\theta-v_i|\}. \label{eq:pab-theta}
\end{gather}
Often we refer to the above concepts not for the entire
point set $\Univ$ but for some subset $\Set \sse \Univ$; in those cases we refer
to the partition $\cI_{\text{all}}(\Set)$, the sum $L_2(\Set)$,  or the distribution $D_2(\Set)$, etc.
Finally, for subset $\Set \sse \Univ$ of points we define:
\begin{enumerate}[label=(\roman*)]
\item Let $\Delta(\Set) := \max_{\bm x,\bm y \in \Set} \|\bx-\by\|^2$ be the
  squared diameter of point set $\Set$. %
\item Call a pair of points $\bm x,\bm y \in \Set$ \emph{far} if
  $\| \bm x-\bm y\|_2^2 \geq \Delta(\Set)/2$, and \emph{close} if
  $\| \bx -\by\|_2^2 < \Delta(\Set)/k^4$.
\item Define the \emph{stretch} of a pair $\bx,\by \in \Set$ to be
  $s_{\bx \by}(\Set) := \| \bx-\by\|_2^2/d_2(\bx,\by)$. Define the \emph{stretch of the
  set} $s(\Set)$ to be the maximum stretch of any \emph{far pair} in $\Set$. 
\end{enumerate}

\subsection{The Algorithm}
\label{sec:kmeans-algorithm}

The process to construct the threshold tree $T$ for \kmeans is
slightly more complex than for \kmed: as before we start off with the
root representing the entire point set $\RR^d$. Now, given a node $v$
representing some box $B_v$ (giving us a subset
$\Univ_v := \Univ \cap B_v$), define the distribution
$D_2(v) := D_2(\Univ_v)$. Now consider the stretch
$s(v) := s(\Univ_v)$ of the set of points.
\begin{enumerate}
\item \emph{(Solo Cuts)} If $s(v) \geq \smash{\frac{|\Univ_v|}{\ln^2 |\Univ_v|}}$, let
  $\bm p^*, \bm q^*$ be a pair of far points in $\Univ_v$ of stretch $s(v)$, and pick a
  threshold cut $(i,\theta) \sim D_2(v)$, conditioned on separating
  the pair $\bm p^*, \bm q^*$. This partitions the box $B_v$ into two
  boxes, and recurse on both.
\item \emph{(Bulk Cuts)} Else if $s(v)$ is smaller than the quantity
  above. In this case, repeatedly sample cuts from $D_2(v)$
  conditioned on not separating any close pairs of points in
  $\Univ_v$, until all pairs of far points in $\Univ_v$ are
  separated. Apply all these cuts in sequence, partitioning $B_v$ into
  potentially multiple pieces; recurse on each of them.
\end{enumerate}
The process stops when each leaf of $T$ contains a single point from
$\Univ$. For the analysis below, we consider a \emph{compressed
  threshold tree} $T'$ which is a tree with branching at least two: if
we perform a solo cut at node $v$, we call it a \emph{solo node}, and
it has two children corresponding to the two parts obtained by this
cut. If we perform bulk cuts at node $v$, it has potentially multiple
children (one for each smaller box obtained by applying these cuts),
and we call it a \emph{bulk node}. 

In the following, let $S(T')$ and $B(T')$ denote the
solo and bulk nodes in $T'$. 
For any node $v \in T'$, let $L_2(v) :=
L_2(\Univ_v)$, and similarly for other parameters defined above.

\subsection{The Expected Cost Increase}

We will bound the cost for the solo cuts and bulk cuts
separately. First we give some preliminary lemmas, then we bound the
cost due to bulk cuts, and finally the cost for solo cuts.

\begin{lemma}[Separation Probability]
  \label{lem:cost_lowerbound}
  For any subset of points $\Set \subseteq \Univ$ and a point
  $\bm p \in \Set$,
  \[ \Pr_{(i,\theta)\sim D_2(\Set)}[ (i,\theta) \text{ separates the origin from
    } \bm p] \leq \frac{2\|\bm p\|_2^2}{L_2(\Set)}\,.\]
\end{lemma}

\begin{proof}
  The probability that a random cut sampled from $D_2(\Set)$ separates the origin $\bm 0$ and $\bm p$ is: 
  \begin{align}
    &\sum_{i \in [d]} \sum_{[a,b] \in \intervals_i} \frac{(a -
      b)^2}{L_2(\Set)}\int_a^b P_{a,b}(\theta) \cdot \mathbbm{1}[\theta
      \text{ is between } 0 \text{ and } p_i]\, d\theta \notag\\
    &= \frac{4}{L_2(\Set)}\sum_{i \in [d]} \sum_{[a,b] \in
      \intervals_i}
      \int_a^b \min(\theta-a, b-\theta) \cdot \mathbbm{1}[\theta
      \text{ is between } 0 \text{ and } p_i]\,d\theta \label{eq:3} \\ 
    & \leq \frac{4}{L_2(\Set)}\sum_{i \in [d]} \sum_{[a,b] \in
      \intervals_i} \int_a^b |\theta-p_i| \cdot
      \mathbbm{1}[\theta \text{ is between } 0 \text{ and }
      p_i]d\theta \label{eq:5} \\ 
    &= \frac{4}{L_2(\Set)}\sum_{i \in [d]} \int_{-\infty}^{\infty}
      |\theta-p_i| \cdot \mathbbm{1}[\theta \text{ is between } 0 \text{
      and } p_i]\,d\theta =  \frac{4}{L_2(\Set)}\sum_{i \in [d]}(p_i)^2/2 =
      \frac{2\|\bm p\|_2^2}{L_2(\Set)}.   \notag
  \end{align}
  Equality~(\ref{eq:3}) and inequality~(\ref{eq:5}) use the expressions
  for $P_{a,b}(\theta)$ given in~(\ref{eq:pab-theta}).  %
\end{proof}

Given a set $\Set$, suppose we sample cuts from distribution
$D_2(\Set)$ with an added rejection step if the cut separates some
pair of points in $\Set$ whose distance is at most $\Delta(\Set)/k^4$.
Formally, let $R(V)\subseteq \cI_{\text{all}}(\Set)$ be the subset of
intervals which are contained in projections of close centers in $V$
onto the coordinate axis. Let
\[ L_2'(V):= \sum_{(i,[a,b])\in \cI_{\text{all}}(\Set)\backslash
  R(V)}|b-a|^2.\] The distribution $D_2'(V)$ picks an interval $[a,b]$
in $\cI_{\text{all}}'(\Set):=\cI_{\text{all}}(\Set)\backslash R(V)$ with probability
$\frac{(b-a)^2}{L_2'(V)}$ and then a cut is chosen from this interval
with the same distribution as $P_{a,b}(\theta)$.
\begin{prop}
\label{cor:rejection_scaling}
For any subset of points $V$, we have $L_2'(V)\geq L_2(V)/2$.
\end{prop}
\begin{proof}
  The sum of squared length of intervals in $R(V)$ is at most the
  total sum of squared distance between all pairs of close centers,
  which is at most
  \begin{align*}
    \binom{|V|}{2}\cdot \frac{\Delta(V)}{k^4}\leq \frac{\Delta(V)}{2k^2}\leq \frac{L_2(V)}{2k}.
  \end{align*}
  This implies that $L_2'(V)\geq L_2(V)(1-1/2k)\geq \frac{1}{2}\cdot L_2(V)$.
\end{proof}

\begin{lemma}[Expected Number of Cuts]
  \label{lem:num-bulk}
  For any node $\Set$, the expected number of cuts from $D_2'(\Set)$
  until all far pairs in $\Set$ %
  are
  separated is at most $24\ln|\Set|\cdot s(\Set) \cdot \frac{L_2(\Set)}{\Delta(\Set)}$.
\end{lemma}

\begin{proof}
  Consider a collection of $M:=3\ln|\Set|\cdot \frac{4s(\Set)\cdot L_2(\Set)}{\Delta(\Set)}$ cuts
  sampled from $D_2'(\Set)$, and consider two ``far'' points $\bm p, \bm q$,
  i.e., such that $\| \bm p - \bm q\|_2^2 > \nf{\Delta(\Set)}2$.  Then, the
  probability that any one cut separates the two is at least
  \begin{align*}
    \Pr_{(i, \theta) \sim D'_2(\Set)} [\mbox{$(i, \theta)$ separates $\bm
    p, \bm q$}] \geq \frac{d_2(\bm p, \bm q)-\binom{k}{2}\frac{\Delta(\Set)}{k^4}}{L_2(\Set)}
      &\geq \frac{\|\bm p - \bm q\|_2^2}{s(\Set) \cdot L_2(\Set)}-\frac{\Delta(\Set)}{2k^2\cdot L_2(\Set)} \\
      &\geq \frac{\Delta(\Set)}{2s(\Set)\cdot L_2(\Set)}- \frac{\Delta(\Set)}{2k^2\cdot L_2(\Set)}\\
      &\geq \frac{\Delta(\Set)}{4s(\Set) \cdot L_2(\Set)}\,.
  \end{align*}
  The last inequality in the above equation follows using
  $s(\Set)\leq k$ and $k\geq 2$. Hence the probability that the $M$
  cuts do not separate some pair at distance at least $\Delta(\Set)/2$
  can be upper-bound using a union bound by
  \begin{align*}
    {{|\Set|} \choose 2} \cdot \left( 1 -\frac{\Delta(\Set)}{4\cdot s(\Set) \cdot L_2(\Set)}\right)^{M} 
    \leq |\Set|^2 \left( \frac{1}{e}\right)^{(3\ln |\Set|) }  = 1/|\Set|\,.
  \end{align*}
  Hence, these $M$ cuts separate all pairs that have
  squared distance at least $\Delta(\Set)/2$ with probability at least
  $1-1/|\Set| \geq \nf12$. In turn, the expected number of cuts is at
  most $2M$.
\end{proof}

We can now start to bound the cost incurred due to bulk cuts.

\begin{lemma}[Logarithmic Number of Relevant Levels] 
For any cut $(i,\theta)$, we have
  \label{lem:log-levels}
  \[ \sum_{v \in B(T')} \ones_{[\bm 0 \in B_v]} \cdot \ones_{[(i,\theta)
    \in \supp(D_2'(v))]} \leq 4\ln k. \]  
\end{lemma}

\begin{proof}
  The bulk nodes in the compressed tree $T'$ that correspond to the part containing
  the origin $\bm 0$ lie on a root-leaf path; call these
  $v_1, v_2, \ldots, v_\ell$, with $v_1$ closest to the root. Our
  algorithm ensures that $\Delta(v_j) \leq
  \Delta(v_{j-1})/2$. Consider the lowest integer $j$ such that
  $(i,\theta)$ belongs to the support of $D_2'(v_j)$, and let
  $\bm p, \bm q \in \Univ \cap B_{v_j}$ be the closest pair of points in
  $B_{v_j}$ separated by $(i,\theta)$. The definition of the
  probability distribution $D_2'(v_j)$ ensures that
  $\|\bm p - \bm q\|_2^2 \geq \Delta(v_j)/k^4$. For
  $j' = j + 4 \ln k + 1$ we have that
  $\Delta(v_{j'}) < \Delta(v_j)/k^4$, and so there are no pairs of
  points separated by $(i,\theta)$---implying that this cut will no
  longer be in the support of $D_2(v_{j''})$ for $j'' \geq j'$.
\end{proof}

\begin{lemma}[Cost for Bulk Cuts]
\label{lem:bulk_costinc}
  The expected cost increase due to bulk cuts is at most $O(k)\cdot \| p^* \|_2^2$.
\end{lemma}

\begin{proof}
  Consider a bulk node $v$ in the decision tree created by the
  algorithm: we generate a random number of bulk cuts $\cuts_v$ at
  this node, and each of these can cause an increase in cost. Let
  $Y_t$ be the following upper bound on the increase in cost due to
  the $t^{\text{th}}$ such cut $(i,\theta)$:
  \[ Y_t := \Delta(v)\cdot \ones_{[(i,\theta) \text{ seps } \bm 0,
    \bm p^* ]}\cdot \ones_{[\{\bm 0, \bp^*\} \sse B_v]} . \label{eq:6} \]
  Moreover, let $N$ be the number of such cuts, then the total
  expected cost is $\EE[\sum_{t = 1}^N Y_t ]$. Since the $Y_t$ variables
  are independent and $N$ is a stopping time, we can use Wald's
  equation to infer that the total expected cost due to these cuts is
  $\EE[N]\cdot\EE[Y_t]$. Taking expectations of the expression for
  $Y_t$ above (with respect to the distribution $D_2'(v)$), and using
  \Cref{lem:num-bulk} to bound $\EE[N]$, this is at most 
  \[ \bigg( O(s(v) \ln|\Univ_v|)\cdot \frac{L_2(v)}{\Delta(v)}\bigg) \cdot \Delta(v)\cdot \sum_{(i,[a,b]) \in
      \intervals_{\all}'(v)} \frac{(b-a)^2}{L_2'(v)}\int_a^b
    P_{a,b}(\theta)\cdot \ones_{[(i,\theta) \text{ seps } \bm 0,
    \bm p^* ]}\cdot \ones_{[\{\bm 0, \bm p^*\} \sse B_v]}
    \,d\theta. \label{eq:6b} \]
 Using \Cref{cor:rejection_scaling}, we know that 
  $L_2'(v)\geq L_2(v)/2$, 
  so the above expression is at most
  \[ O(s(v) \ln|\Univ_v|) \cdot \sum_{(i,[a,b]) \in
      \intervals_{\all}(v)} (b-a)^2 \int_a^b
    P_{a,b}(\theta)\cdot \ones_{[(i,\theta) \text{ seps } \bm 0,
    \bm p^* ]}\cdot \ones_{[\{\bm 0, \bm p^*\} \sse B_v]} \cdot \ones_{[(i,\theta)
    \in \supp(D_2'(v))]}
    \,d\theta. \label{eq:6c} \]
  Next, we observe that for any dimension $i$, we have
  \[ (b-a)^2 \cdot P_{a,b}(\theta) \cdot \ones_{[\bm p^* \in B_v]} \leq 4
    | \theta - \bm p_i^*|. \] Moreover, for each bulk node we have
  $s(v) \ln |\Univ_v| \leq |\Univ_v|/\ln |\Univ_v|$. This in turn is at most
  $k/\ln k$, since the function $\frac{x}{\ln x}$ is monotone and
  $|\Univ_v| \leq k$.  Substituting both these facts, we get
  \[ O(\nf{k}{\ln k}) \cdot \sum_i 
    \int_{-\infty}^\infty
    |\theta - \bm p_i^*| \cdot \ones_{[(i,\theta) \text{ seps } \bm 0,
    \bm p^* ]}\cdot \ones_{[\bm 0 \in B_v]} \cdot \ones_{[(i,\theta)
    \in \supp(D_2'(v))]}
    \,d\theta. \label{eq:6d} \]
  Next, we use \Cref{lem:log-levels} to get:
  \[ \sum_{v \in B(T')} \ones_{[\bm 0 \in B_v]} \cdot \ones_{[(i,\theta)
    \in \supp(D_2'(v))]} \leq O(\ln k). \]
  Now summing over all $v$, we get
  \[ O(k) \cdot \sum_i 
    \int_{-\infty}^\infty
    |\theta - \bm p_i^*| \cdot \ones_{[(i,\theta) \text{ seps } \bm 0,
    \bm p^* ]}
    \,d\theta = O(k) \cdot \|\bm p^*\|_2^2. \]
  This completes the proof.
\end{proof}

Finally, we turn our attention to solo cuts.
For solo node $v$ let $\bm p_v, \bm q_v$ be the two far nodes such
that their stretch is at least $\frac{|\Univ_v|}{(\ln |\Univ_v|)^2}$,
and define the distribution $D_2''(v)$ to be the distribution $D_2(v)$
conditioned on separating this far pair. 

\begin{lemma}[Ratio for Solo Cuts]
  \label{lem:solo-balance-lemma}
  For any solo node $v \in T'$, 
  \[
    \frac{\EE_{(i,\theta) \sim D_2''(v)}[\text{ cost increase at node
      $v$ } ]}{ \EE_{(i,\theta) \sim D_2''(v)}[\text{ size of smaller
        child of node $v$ }]} \leq 32\|\bm p^*\|_2^2\cdot
    \left(1+\ln\left(\frac{|\Univ_v|}{s(v)}\right) \right).
  \]
\end{lemma}

\begin{proof}
  \Cref{lem:cost_lowerbound} implies that
  \[ \EE_{(i,\theta) \sim D_2(v)}[ \text{ cost increase } ] \leq \Delta(v)
    \cdot \frac{2\|\bm p^*\|_2^2}{L_2(v)}. \]
  Moreover, the probability of separating $\bm p_v, \bm q_v$ is at least
  \[ \frac{d_2(\bm p_v, \bm q_v\mid \Univ_v)}{L_2(v)}
    \geq \frac{\|\bm p - \bm q\|_2^2}{s(v) \cdot L_2(v)}
    \geq \frac{\Delta(v)}{2\cdot s(v) \cdot L_2(v)}\,. \]
  Since the cost increase is non-negative, we get that
  \begin{gather}
    \EE_{(i,\theta) \sim D_2''(v)}[ \text{ cost increase } ] \leq
    \Delta(v) \cdot \frac{2\|p^*\|_2^2}{L_2(v)} \cdot \frac{2\cdot s(v)
      \cdot L_2(v)}{\Delta(v)} \leq 4 s(v) \cdot \|\bm p^*\|_2^2. \label{eq:7}
  \end{gather}
  This bounds the numerator of the desired quantity; for the
  denominator, we prove the following claim in \Cref{sec:proof-stretch-vs-separation}:
  \begin{restatable}{claim}{StretchSep}
    \label{clm:stretch-separation}
    If for a cut $(i,\theta)$, we define $H^+ := \{x \mid x_i \geq \theta\}$ and $H^- = \RR^d \setminus
    H^+$, then 
    \[ \EE_{(i,\theta) \sim D_2''(v)}\bigg[ \min\big( |\Univ_v \cap H^+|, |\Univ_v \cap
      H^-|\big) \bigg] \geq \frac{s(v)}{8(1+\ln(\nicefrac{|\Univ_v|}{s(v)}))}~. \]
  \end{restatable}
  Using \Cref{clm:stretch-separation} with~(\ref{eq:7}) finishes the proof.
\end{proof}

\begin{lemma}[Cost for Solo Cuts]
\label{lem:solo_costinc}
  For any internal solo node $v \in T'$, the expected cost increase due
  to solo cuts made in the subtree $T'_v$ is at most \[32|\Univ_v|(1+ 2\ln
  \ln |\Univ_v|) \cdot \ones_{[\bm 0 \in B_v]} \cdot \| \bm p^* \|_2^2. \]  
\end{lemma}
\begin{proof}
  The proof is by induction. The base cases are when node $v$ has
  $|\Univ_v| \leq 3$. For $v$ to be an internal node,
  $|\Univ_v| \in \{2,3\}$. If it has two nodes, then
  $s(v) = 1 < \frac{2}{\ln^2 2}$, and hence $u$ cannot be a solo
  node. If we have $3$ points, then the only solo node in the subtree
  $T'_v$ is $u$ itself, so it suffices to argue that the expected cost
  increase due to this solo cut is at most
  $32|\Univ_v|(1+2\ln \ln |\Univ_v|)\cdot \|\bm p^*\|_2^2$. From
  (\ref{eq:7}) we know that the expected cost increase is at most
  $4s(u) \cdot \|\bm p^*\|_2^2$. Now using $s(u)\leq |\Univ_v|$ and
  $|\Univ_v|\leq 3$, we obtain the required bound.

  Else consider some node $v$ with $|\Univ_v| \geq 4$, and let
  $\chi(v)$ be the set of solo nodes whose closest solo ancestor is
  $v$: it follows that $\sum_{w\in \chi(v)} |\Univ_w| \leq
  |\Univ_v|$. Moreover, suppose the random solo cut
  $(i_v, \theta_v) \sim D_2''(v)$ partitions $\Univ_v$ into parts of size
   at least $\sigma(v)$, then each
  $|\Univ_v| - |\Univ_w| \geq \sigma(v)$. Using the induction hypothesis
  on each $w \in \chi(v)$, the total expected cost increase is at most 
  \begin{gather}
    \EE[ \text{ cost increase at $v$ }] + \sum_{w \in \chi(v)} 32|\Univ_w|(1+2\ln
    \ln |\Univ_w|) \cdot \ones_{[\bm 0 \in B_w]} \cdot \| \bm p^* \|_2^2.  
  \end{gather}
  Since the origin belongs to at most one of the sets $B_{w^*}$, the
  sum contributes at most
  \begin{gather}
    32|\Univ_{w^*}|(1+2 \ln \ln |\Univ_v|) \cdot \| \bm p^* \|_2^2 \leq
    32(|\Univ_{v}| - \EE[\sigma(v)])(1+2 \ln \ln |\Univ_v|) \cdot \| \bm p^*
    \|_2^2.
  \end{gather}
  (The RHS is non-negative, since $\sigma(v) \leq |\Univ_v|$, so the bound holds even when $\chi(v) = \emptyset$.)
  Now \Cref{lem:solo-balance-lemma} implies that the cost at $v$ is at most
  \[  \EE[ \sigma(v) ] \cdot 32(1+\ln(\nicefrac{|\Univ_v|}{s(v)}))  \cdot
    \|\bm p^*\|_2^2 . \]
  Finally, using that $s(v) \geq |\Univ_v|/(\ln |\Univ_v|)^2$ for a
  solo node, and summing the two terms, completes the proof.
\end{proof}
We can wrap up: Using \Cref{lem:bulk_costinc} and \Cref{lem:solo_costinc}, we know that the expected cost increase due to all the cuts used to separate the points in $U$ is at most $O(k\ln\ln k)\cdot \|\bm p^*\|_2^2$. If $\widehat{\bm p}$ is the unique point in $U$ in the region $B_u$ corresponding to the leaf node $u$ such that $\bm 0 \in B_u$, then the cost is upper bounded by 
\begin{align}
\label{eqn:kmeans_conclusion}
\|\widehat{\bm p}\|_2^2\leq 2\cdot \|\widehat{\bm p}-\bm p^*\|_2^2+2\cdot \|\bm p^*\|_2^2
\end{align}
using the generalized triangle inequality. Taking expectation on both sides of \Cref{eqn:kmeans_conclusion} and plugging in $\|\bm p^*\|_2^2=1$ gives $f_2(U)\leq O(k\ln \ln k)$.

\section{Tight Bounds for the Random Threshold Algorithm}
\label{sec:tight}

We now improve the bound of $(1+o(1)) \ln k$ from
\S\ref{sec:rand-threshold} to give an exact bound of $1 + H_{k-1}$ for
the \RT algorithm; we first show an example which achieves this bound,
and then give our precise analysis of the algorithm.

\subsection{A Lower Bound}
\label{sec:tight-analysis}

Consider an instance in $\RR^k$ given by a reference clustering having
one ``close'' center $\bmu^1 = e_1$, and $k-1$ ``far'' centers
$\bmu^i = M\, e_i$ for each $i \in \{2, \ldots, k\}$, where scalar $M \gg 1$. We
consider a single data point at the origin. (As always, we can imagine
there being many points colocated with each of the centers.) Let the
expected assignment cost due to the algorithm for the point at the
origin on an instance with $j$ far points be denoted by $g(j)$.  Then
we get a recurrence:
\[ \textstyle g(k) = \frac{1}{M(k-1)+1} \cdot M + \frac{M(k-1)}{M(k-1)+1} \cdot
  g(k-1), \]
and $g(\bm 0) = 1$. As $M \to \infty$, this gives us $g(k) \to 1 + H_{k-1}$,
as claimed. %
In the next section, we will prove a matching upper bound of
$1+H_{k-1}$. 

\subsection{Towards a Matching Upper Bound}
\label{sec:matching-upper-bound}

Our proof for this case is  technical, so the reader may want to
 keep three special cases in mind:
\begin{enumerate}
\item The ``axis-aligned'' case, inspired by the bad example: the
  points in $\Univ$ are $\{ \bp_i := d_i \bm e_i\}$, where $1 = d_1 \leq
  d_2 \leq \ldots \leq d_k$, 
\item the ``orthogonal closest-point'' case, where the closest point
  is $\bm e_1$, and all other points lie in the orthogonal subspace to
  it, and
\item the ``uniform'' case, where all points other than the closest
  are at the same distance $d \geq 1$.
\end{enumerate}
Many of our proofs become simpler in these special cases, and thinking
about these cases will give us crucial intuition. 

We start with a set of points $U \sse \RR^d$, recall that $f(U)$ was
defined to be $\EE[ \|\bf \widehat{p}\|_1 ]$, where $\bf \widehat{p}$ is the
unique point in the box containing the origin in the \RT algorithm. As
in \S\ref{sec:clocks-cost}, assume we have a dimension for each cut
$S \sse U$, and point $\bp \in U$ has value $z_S \ones[\bp \in S] \geq 0$
in this coordinate. Finally defining $\cC_S$ to be the collection of sets
that cross $S \sse U$, we get that for any $S \sse U$,
\begin{align}
	\label{eqn:definition_global}
	f(S)= \frac{\sum_{E\in \cC_S}z_E\cdot f(S\setminus
  E)}{\sum_{E \in \cC_S} z_E}, \qquad \text{and} \qquad  f(\{\bp\})=\|\bp\|_1.
\end{align}
Moreover, when $\sum_{E \in \cC_S} z_E = 0$, the value of $f(S) = 0$.
Let us denote by $\ell(\bp):= \|\bp\|_1$ and
$\beta_{k-1}:= 1+H_{k-1}$.  Using \Cref{eqn:definition_global}, we
think of the function $f(U)$ purely as an algebraic function of the
$z_S$ values, and our central goal in this section to prove the following:
 \begin{theorem}[Main Goal]
   \label{goal:goal_1}	
   For any point $\bp\in U$, the value $f(U)\leq \beta_{k-1}\cdot\ell(\bp)$.
\end{theorem}

Since the ratio $f(U)/\ell(\bp)$ is difficult to argue about, we
instead focus on bounding the derivative
$\frac{\partial f(U)}{\partial z_E}$ by $\beta_{k-1}$.  The following
lemma shows that, by integrating along a path from the origin to the
point $\bp$, such a bound on the derivative suffices: (a formal proof
appears in \S\ref{sec:app-tight}) 
\begin{lemma}
  \label{lem:reduce-to-deriv}	
  For any $\bp \in U$, if
  $\frac{\partial f(U)}{\partial z_E} \leq \beta_{k-1}$ for all
  $E \sse U$ with $\bp \in E$, then
  $f(U)\leq \beta_{k-1}\cdot\ell(\bp)$.
\end{lemma}

\subsection{Bounding the Derivative}
\label{sec:bound-deriv}

We start by taking definition~(\ref{eqn:definition_global}) and
calculating the derivative $\frac{\partial f(S)}{\partial z_T}$ for
the case $|S| \geq 2$:
\begin{align}
	\label{eqn:derivative}	\frac{\partial f(S)}{\partial z_T} = \frac{\sum_{E\in \cC_S}z_E\cdot \frac{\partial f(S\setminus E)}{\partial z_T}+ \mathbbm{1}[T\in \cC_S]\cdot \big(f(S\setminus T)-f(S)\big)}{\sum_{E\in \cC_S}z_E}.
\end{align}
When $|S|=1$, if $S=\{r\}$, then  $f(S)=\ell(r)$ by definition, and so $\frac{\partial f(S)}{\partial z_T}=\frac{\partial \ell(r)}{\partial z_T}$.
Henceforth, let us fix a set $S \sse U$ and some subset $T$.  %
The next lemma (in \S\ref{sec:app-tight}) follows by direct calculations: 
\begin{lemma} The partial derivatives satisfy:
  \label{lem:derivative_observations}
  \begin{enumerate}[nosep,label=(\roman*)]
  \item 	\label{lem:shift_derivative}
    If $T \supseteq S$, then $\frac{\partial f(S)}{\partial z_T}=1$.
  \item 	\label{lem:empty_derivative}
    If $T \cap S = \emptyset$, then $\frac{\partial f(S)}{\partial z_T}=0$. 
  \item 	\label{cor:shift_inequality}
    We have $f(S)\geq \sum_{E\supseteq S }z_E$.
  \end{enumerate}
\end{lemma}

The partial derivative $\frac{\partial f(S)}{\partial z_T}$ is not
very well-behaved: e.g., it is not guaranteed to be non-negative
(which turns out to make an inductive proof difficult). To address this
issue, we define a surrogate \emph{pseudo-derivative} operator, and use
bounds on this pseudo-derivative to bound the derivative.

\begin{defn}[Pseudo-Derivative]
  \label{defn:pseudo_derivative} The pseudo-derivative of $f(S)$ with
  respect to the variable $z_T$ such that $T$ crosses $S$ (i.e., $T
  \in \cC_S$) is:
  \begin{align}
    \label{eqn:pseudo_derivative}
    \frac{\widehat{\partial} f(S)}{\widehat{\partial} z_T}
    &= %
    \frac{\sum_{E \in \cC_S
      }z_E\cdot {\color{purple} \frac{\widehat{\partial} f(S\setminus
          E)}{\widehat{\partial} z_T}}+ f(S\setminus T)-
      {{\color{purple} \sum_{E\supseteq S }z_E}}}{\sum_{E\in \cC_S}z_E}.
  \end{align}
        It is defined to be $1$ if $T \supseteq S$, and $0$ if $T \cap S = \emptyset$.
\end{defn}
Observe the differences with~(\ref{eqn:derivative}), which are marked
in red: when $T\in \cC_S$, each smaller derivative term
$\frac{\partial f(S\setminus E)}{\partial z_T}$ in the numerator of
the derivative is naturally replaced with the corresponding
pseudo-derivative term
$\frac{\smash{\widehat{\partial} f(S\setminus E)}}{\widehat{\partial}
  z_T}$, but crucially, the term $f(S)$ is replaced with
$\sum_{E\supseteq S}z_E$. These are the terms corresponding to the
``shift'' of the set $S$---i.e., the cuts that separate all points in
$S$ from the origin---and hence they form a lower bound on $f(S)$, the
expected distance to the closest point in $S$. This latter change
makes the following arguments easier (and indeed, possible), but still
maintains the intuition of the derivative being invariant under
translations. 
In \S\ref{sec:app-tight}, we prove the following lemma, showing it is
indeed an upper bound.

\begin{lemma} 
  \label{lem:pseudoderivative_upperbound}
  The pseudo-derivative is non-negative, and bounds the derivative from above. I.e.,
  \begin{align*}
    \max \left(	\frac{\partial f(S)}{\partial z_T} , 0 \right) \leq 	\frac{\widehat{\partial} f(S)}{\widehat{\partial} z_T}.
  \end{align*}
\end{lemma}

Given \Cref{lem:pseudoderivative_upperbound}, it suffices to upper
bound the pseudo-derivative by $\beta_{k-1}$, which we do
next.
\begin{theorem}
  \label{lem:final}
  For any $S\subseteq U$ and any $T\neq \emptyset$, we have
  $\frac{\widehat{\partial} f(S)}{\widehat{\partial} z_T} \leq \beta_{|S\setminus T|}.$
\end{theorem}
\Cref{lem:final} implies our desired bound on the derivative, because
$|S\setminus T|\leq k-1$ for any $S\subseteq U , T\neq \emptyset $.
To prove \Cref{lem:final}, we first prove it for the special case when
all points in $S\setminus T$ have the same norm (the \emph{uniform}
case), and then reduce the general case to this uniform case.

\subsubsection{Proof of \Cref{lem:final}: the Uniform Case}
\label{subsection:special_case}
The main reason that it is easier to prove \Cref{lem:final} for the uniform case is because we know the value of $f(S\setminus T)$ exactly which will be equal to the norm of all the points in $S\setminus T$. Otherwise, it is hard to obtain any upper bound to $f(S\setminus T)$ in the general case. Moreover, we know that the uniform property holds true for all subsets of $S$. This enables us to use the upper bound of \Cref{lem:final} to derivative terms $\frac{\widehat{\partial} f(S\setminus
	E)}{\widehat{\partial} z_T}$ by $\beta_{|S\setminus (E\cup T)|}$.
\begin{lemma}
	\label{lem:special_case}
	If all points in $S\setminus T$ have the same norm, then $\frac{\widehat{\partial} f(S)}{\widehat{\partial} z_T} \leq \beta_{|S\setminus T|}$.
\end{lemma}
\begin{proof}
	The proof is by induction on $|S|$. If $T\notin \cC_S$, we know that $	\frac{\widehat{\partial} f(S)}{\widehat{\partial} z_T}$ is either $0$ or $1$. But $\beta_{|S\setminus T|}\geq 1$ as $\beta_{m}\geq 1$ for any $m\geq 0$ which implies that we are done. Hence, for $|S|=1$, we know that $T\notin \cC_S$ and we are done. From now on, we can assume $|S|\geq 2$ and $T\in \cC_S$. So we use the definition of the pseudo-derivative from (\ref{eqn:pseudo_derivative})
	and use the upper bound of the lemma inductively for the
        recursive terms $\frac{\widehat{\partial} f(S\setminus
          E)}{\widehat{\partial} z_T}$ to get:
	\begin{align}
		\label{eqn:upperbound_pseudoderivative}
		\frac{\widehat{\partial} f(S)}{\widehat{\partial} z_T}\leq  	\frac{\sum_{E\in \cC_S }z_E \cdot \beta_{|S\setminus(E\cup T)|}+ f(S\setminus T)- \sum_{E\supseteq S}z_E}{\sum_{E\in \cC_S }z_E}.
	\end{align}
	In order to show that $	\frac{\widehat{\partial} f(S)}{\widehat{\partial}
		z_T} \leq \beta_{|S\setminus T|}$, it is sufficient
	to upper bound the RHS of \Cref{eqn:upperbound_pseudoderivative} by $\beta_{|S\setminus T|}$. Simplifying gives the following sufficient condition
	\begin{align}
		\label{eqn:6}	f(S\setminus T)\leq \sum_{E\supseteq S }z_E + \sum_{E\in \cC_S}z_E\left(\beta_{|S\setminus T|}- \beta_{|S\setminus (T\cup E)|} \right).
	\end{align}
	From here on, we will prove \Cref{eqn:6}. Since $T\in \cC_S$, we know that $|S\setminus T|\geq 1$. All points in $S\setminus T$ have the same norm, so we can write 
	\begin{align}
		\label{eqn:11}  f(S\setminus T)&=	\frac{1}{|S\setminus T|}\sum_{r\in S\setminus T}\ell(r) \\
		\label{eqn:12} &\stackrel{\eqref{eqn:norm_expansion}}{=} \frac{1}{|S\setminus T|}\sum_{r\in S\setminus T}\sum_{E:r \in E}z_E \\
		\label{eqn:13} &= \sum_{E}z_E\cdot \frac{|(S\setminus T)\cap E|}{|S\setminus T|}.
	\end{align}
	When $E\supseteq S$, the coefficient of $z_E$ in \Cref{eqn:6} is $1$. The coefficient of $z_E$ in \Cref{eqn:13} is also $1$ because in this case, $S\subseteq E$ and hence $|(S\setminus T)\cap E|=|S\setminus T|$. Otherwise, the coefficients of $z_E$ in \Cref{eqn:6} and \Cref{eqn:13} are $\beta_{|S\setminus T|}- \beta_{|S\setminus (T\cup E)|} $ and $|(S\setminus T)\cap E|/|S\setminus T|$ respectively. It remains to show
	\begin{align}
		\label{eqn:8}	\frac{|(S\setminus T)\cap E|}{|S\setminus T|} \leq 	\beta_{|S\setminus T|}- \beta_{|S\setminus (T\cup E)|}.
	\end{align}
	We can justify \Cref{eqn:8} because the left hand side is sum of $|(S\setminus T)\cap E|$ copies of $\frac{1}{|S\setminus T|}$, whereas the right hand side is the sum of $|(S\setminus T)\cap E|$ terms, the smallest of them equal to $\frac{1}{|S\setminus T|}$.
\end{proof}
\begin{cor}
	\label{lem:base_case}
	If $T \in \cC_S$ and $|S\setminus T|=1$, we have $\frac{\widehat{\partial} f(S)}{\widehat{\partial} z_T}\leq \beta_{1}=2$. 
\end{cor}
\begin{proof}
	Follows from \Cref{lem:special_case} because when there is only one point in $S\setminus T $, we can say that all points in $S\setminus T$ have the same norm.
\end{proof}

\subsection{Proof of \Cref{lem:final}: Reducing to the Uniform Case}

The case for points in $\Univ \backslash T$ have different norms is
the technical heart of the proof. In this case, we ``lift'' the points
in such a way that the value of the pseudo-derivative
$\frac{\hat{\partial}f(S)}{\hat{\partial}z_T}$ is monotonically
increasing, thereby reducing to the uniform case of
\Cref{subsection:special_case}. To begin, we give a few supporting
lemmas.
\subsubsection{Supporting lemmas}
\begin{obs}
	\label{obs:upperandlowerbound}
	$\max\limits_{b\in S}(\ell(b))\geq f(S)\geq \min\limits_{a\in S}(\ell(a))$
\end{obs}
\begin{proof}
 The proof follows from the fact that $f(S)$ is an expectation of norms of points in $S$.
\end{proof}
Let us denote $S\backslash \{\bp\}$ simply as $S-\bp$. 
\begin{obs}
	\label{obs:boundary_point}
	For any point $\bp\in S$ such that $|S|\geq 2$, if $E \in \cC_S$ and $E \notin \cC_{S-\bp}$, then $E\cap S$ is equal to either $\{\bp\}$ or $S-\bp$. 
\end{obs}
\begin{obs}
	\label{obs:extended_formula}
	Let $A = \sum_{j=1}^m w_j\cdot v_j/\sum_{j=1}^m w_j$ be a weighted average of $m$ items $t_1,\dots, t_m$ with
	values $v_1,\dots,v_m$ weighted by respective weights
	$w_1,\dots, w_m$. Adding any number $m'-m$ items of value $A$
	(with any weights) does not change the weighted average. I.e., the new
	weighted average is
	\begin{align*}
		A'&=\frac{\sum_{j=1}^{m}w_{j}\cdot v_{j}+
			\sum_{j=m+1}^{m'}w_{j}\cdot A}{\sum_{j=1}^{m}w_j + \sum_{j=m+1}^{m'}w_{j}}
		= A
	\end{align*}
\end{obs}
\Cref{obs:extended_formula} is simple but powerful. 
\begin{lemma}
	\label{lem:delete_min_element}
	Let $\bp$ be a point in $S$ having the minimum $\ell_1$ norm. If $S -\bp \neq \emptyset$, then
	\begin{align}
		\label{eqn:delete_min_element}
		f(S-\bp)\geq f(S) \geq f(\{\bp\})
	\end{align}
\end{lemma}
\begin{proof}
	The lower bound follows from \Cref{obs:upperandlowerbound}. Let us prove the upper bound by induction in $|S|$. If $|S|=2$, then $f(S-\bp)$ is the norm of the maximum norm point in $S$ which is more than $f(S)$ by \Cref{obs:boundary_point}. If $|S|\geq 3$, then $|S-\bp|\geq 2$ so we will use \Cref{eqn:definition_global} to expand both $f(S)$ and $f(S-\bp)$ as 
 \begin{align}
   f(S-\bp)=\frac{\sum_{E\in \cC_{S-\bp}}z_E\cdot f((S-\bp)\backslash E)}{\sum_{E\in \cC_{S-\bp}}z_E},\,
   f(S)=\frac{\sum_{E\in \cC_{S}}z_E\cdot f(S\backslash E)}{\sum_{E\in \cC_{S}}z_E}.
 \end{align}
 Using \Cref{obs:extended_formula} and \Cref{obs:boundary_point}, we can re-write $f(S-\bp)$ as 
 \begin{align}
 f(S-\bp)=\frac{\sum_{E\in \cC_{S-\bp}}z_E\cdot f((S-\bp)\backslash E)+z_{\{\bp\}}\cdot f(S-\bp)+z_{S-\bp}\cdot f(S-\bp)}{\sum_{E\in \cC_{S}}z_E}.
 \end{align}
 In the numerators of $f(S-\bp)$ and $f(S)$, the coefficient of $z_E$ in $f(S-\bp)$ is more than that of $f(S)$ either by induction on $|S|$ ($S\leftarrow S\backslash E$) or using the fact that $f(S-\bp)\geq f(\{\bp\})$ from \Cref{obs:upperandlowerbound}. 
\end{proof}
From now, many other lemmas will have their proofs very similar to that of \Cref{lem:delete_min_element} where we compare the function values of two different sets by expanding the functions and using \Cref{obs:extended_formula} to normalize the denominators and then use induction and other arguments to conclude. 
\subsubsection{The lifting operation}
Now let us define the lift operation. For a set $S'=\{\bp_1,\dots,\bp_{k'}\}\subseteq \Univ$, consider the following trajectory for the $\bm z$ vector 
\begin{align}
\label{eqn:liftingtrajectory}
 z_E(t):= \begin{cases} 
  z_E+ t & E=\{\bp\},\, \bp \in S' \\
  z_E & \text{ otherwise }.
\end{cases} 
\end{align}
Given this trajectory for the embedding, we are interested in how this changes the function value. Let us define 
\[\frac{\partial f(S)}{\partial S'}:= \frac{\mathrm{d}f_t(S)}{\mathrm{d}t}\bigg|_{t=0}\,.\]
\begin{lemma}
	\label{lem:lift_derivative}
	Let $S'$ be a subset of $\Univ$. Then,
	\begin{enumerate}[nosep]
		\item \label{eqn:disjoint_lifting} If $S\cap S'=\emptyset, \frac{\partial f(S)}{\partial S'}=0$.
		\item \label{eqn:single_point_lifting}If $|S|=1$, $S\subseteq S',  \frac{\partial f(S)}{\partial S'}=1$.
		\item \label{eqn:general_lifting}If $|S| \geq 2$, 	
		\begin{align}
			\label{eqn:derivativeS_simpler}	\frac{\partial f(S)}{\partial S'} = \frac{\sum_{E\in \cC_S}z_E\cdot \frac{\partial f(S\backslash E)}{\partial S'}+ \sum_{ \bp_{i}\in S\cap S'}\left(f(S-\{\bp_{i}\}) - f(S)\right)}{\sum_{E\in \cC_S}z_E}.
		\end{align}	
		Moreover, if points in $S'\cap S$ if any, have the least norm out of points in $S$, then
		$\frac{\partial f(S)}{\partial S'}\geq 0.$
		\end{enumerate}
\end{lemma}
\begin{proof}
Since $f(S)$ is purely a function of points in $S$, and hence only depends on variables $z_E$ such that $\bp\in E$ for some $\bp\in S$. But the only variables that change with $t$ are $z_{\{\bp'\}}, \bp'\in S'$. Since $S'\cap S=\emptyset$, we can conclude statement \ref{eqn:disjoint_lifting}. When $S=\{\bp\}$ and $\bp\in S'$, $f_{t}(S)=\ell_t(\bp)$. Since $z_{\{\bp\}}$ is the only variable that is changing and $\ell(\bp)$ contains this, we can conclude statement \ref{eqn:single_point_lifting}. Once we know that $|S|=2$, using \Cref{eqn:definition_global} and applying the $\frac{\partial}{\partial S'}$ operator on both sides gives 
\begin{align}
\label{eqn:derivative_general_intermediate}
\frac{\partial f(S)}{\partial S'}= \frac{ \sum_{E\in \cC_S}\left(z_E\cdot \frac{\partial f(S\backslash E)}{\partial s'}+\frac{\partial z_E}{\partial S'}\cdot f(S\backslash E)\right )}{\sum_{E\in \cC_S}z_E}-\frac{f(S)}{\sum_{E\in \cC_S}z_E}\cdot \frac{\partial \sum_{E\in \cC_S}z_E}{\partial S'}. 
\end{align}
Using the fact that the only variables that are changing with $t$ are $z_{\{\bp\}}$ for $\bp\in S'$ in \Cref{eqn:derivative_general_intermediate} gives \Cref{eqn:derivativeS_simpler}. Let us argue that $\frac{\partial f(S)}{\partial S'}\geq 0$ when points in $S\cap S'$ have the least norm of points in $S$ using induction on $|S|$. For $|S|=1$, using statement \ref{eqn:single_point_lifting} or statement \ref{eqn:disjoint_lifting}, we are done. Otherwise, using \Cref{eqn:derivativeS_simpler}, we have that the recursive derivative terms $\frac{\partial f(S\backslash E)}{\partial S'}$ in the numerator of \Cref{eqn:derivativeS_simpler} are non-negative by inductive hypothesis. From \Cref{lem:delete_min_element}, we know that $f(S-\bp_i)-f(S)\geq 0$ for every $\bp_i \in S\cap S'$. Combining both these observations, we can conclude statement \ref{eqn:general_lifting} and hence the lemma. 
\end{proof}
\begin{lemma}
	\label{lem:lift_augmentation}
	If $S'$ is the set of points in $\Univ\backslash T$ of minimum norm, then
	\begin{align*}
		\frac{\partial }{\partial S'}\left(\frac{\hat{\partial} f(S)}{\hat{\partial}z_T} \right) \geq 0.
	\end{align*}
\end{lemma}
\begin{proof}
	The proof is by induction on $|S|$. If $T \notin \cC_S$, this is trivially true because $\frac{\hat{\partial} f(S)}{\hat{\partial}z_T}$ is either $0$ or $1$. So for $|S|=1$, we are done. Otherwise, we have $|S|\geq 2$ and $ T\in \cC_S$. Recall from \Cref{eqn:pseudo_derivative} that 
	\begin{align*}
		\frac{\hat{\partial} f(S)}{\hat{\partial} z_T} = \frac{\sum_{E \in \cC_S}z_E\cdot \frac{\hat{\partial} f(S\backslash E)}{\hat{\partial}z_T}+ f(S\backslash T)-\sum_{E\supseteq S }z_E}{\sum_{E\in \cC_S}z_E}.
	\end{align*}
	Applying the operator $\frac{\partial }{\partial S'}$ on both sides gives
	\begin{align*}
		\frac{\sum_{E\in \cC_{S}}z_{E}\cdot \frac{\partial }{\partial S'}\left(\frac{\hat{\partial} f(S\backslash E)}{\hat{\partial}z_T} \right) + \frac{\partial f(S\backslash T)}{\partial S'}- \frac{\partial }{\partial S'}\left(\sum_{E\supseteq S}z_E \right) + \sum_{\bp\in S'\cap S}\left( \frac{\hat{\partial} f(S-\bp)}{\hat{\partial} z_T}- \frac{\hat{\partial} f(S)}{\hat{\partial} z_T}\right) }{\sum_{E\in \cC_{S}}z_E}. 
	\end{align*}		
	Since $|S|\geq 2$, the term $\frac{\partial }{\partial S'}\left(\sum_{E\supseteq S}z_E \right)$ is zero. Using this, what remains is 
	\begin{align*}
		\frac{\sum_{E\in \cC_S}z_E\cdot	\frac{\partial }{\partial S'}\left(\frac{\hat{\partial} f(S\backslash E)}{\hat{\partial}z_T} \right) + \frac{\partial f(S\backslash T)}{\partial S'}+ \sum_{\bp \in S \cap S'}\left( \frac{\hat{\partial} f(S-\bp)}{\hat{\partial} z_T}- \frac{\hat{\partial} f(S)}{\hat{\partial} z_T}\right) }{\sum_{E \in \cC_S}z_E}. 
	\end{align*}
	This means, it is sufficient to show that when $T\in \cC_S$, 
	\begin{align}
		\label{eqn:final_blow}
		\frac{\partial f(S\backslash T)}{\partial S'}+ \sum_{\bp\in S\cap S' }\left( \frac{\hat{\partial} f(S-\bp)}{\hat{\partial} z_T}- \frac{\hat{\partial} f(S)}{\hat{\partial} z_T}\right) \geq 0.
	\end{align}
	Because, the recursive derivatives $\frac{\partial }{\partial
          S'}\left(\frac{\hat{\partial} f(S\backslash
            E)}{\hat{\partial}z_T} \right)$ are at least zero by
        induction on $|S|$. The proof of \Cref{eqn:final_blow} is now
        given in \Cref{lem:final_blow}.
\end{proof}
\begin{lemma}
	Let $S'=\{\bp_1,\dots,\bp_{k'}\}$ be the set of points in $\Univ\backslash T$ of minimum norm. For  $T \in \cC_S$, 
	\label{lem:final_blow}
	\begin{align*}
		\frac{\partial f(S\backslash T)}{\partial S'}+ \sum_{\bp\in S\cap S'}\left( \frac{\hat{\partial} f(S-\bp)}{\hat{\partial} z_T}- \frac{\hat{\partial} f(S)}{\hat{\partial} z_T}\right) \geq 0.
	\end{align*}
\end{lemma}
\begin{proof}
	The proof is by induction on $|S|$. Let us work out boundary cases first. If $S\cap S' = (S\backslash T) \cap S' =\emptyset$, then $\frac{\partial f(S\backslash T)}{\partial S'}=0$ and the summation is empty which makes the entire expression equal to $0$ in which case, we are fine. So from now, we can assume $S \cap S' \neq \emptyset$. If $|S\backslash T|=1$,  then $S\backslash T = \{\bp^{*}\} \subseteq S'$, which implies $\frac{\partial f(S\backslash T)}{\partial S'}=1$. The expression in this case is simply 
	\begin{align*}
		1+\frac{\hat{\partial} f(S- \bp^*)}{\hat{\partial} z_T}- \frac{\hat{\partial} f(S)}{\hat{\partial} z_T} =2-\frac{\hat{\partial} f(S)}{\hat{\partial} z_T}
	\end{align*}
	because $S-\bp^*\in T$. For this boundary case, it remains to prove that $\frac{\hat{\partial} f(S)}{\hat{\partial} z_T}\leq 2$. This is taken care of \Cref{lem:base_case}. For the base case $|S|=2$, we have $|S \backslash T|=1$ in which case, we are done from the preceding arguments. Otherwise, if $|S\backslash T|\geq 2$ we can expand all the derivative terms. But before we do that, we need \Cref{obs:extended_formula} so that we have the same denominator for all the three derivatives. Now using \Cref{eqn:derivativeS_simpler} and the fact that $(S\backslash T)\cap S'= S\cap S'$,
	\begin{align}
		\frac{\partial f(S\backslash T)}{\partial S'}\nonumber &=  \frac{\sum_{E\in \cC_{S\backslash T}}z_E\cdot \frac{\partial f(S\backslash( T\cup E))}{\partial S'}+ \sum_{\bp\in S\cap S'}\left(f(S\backslash (T\cup \{\bp\})) - f(S\backslash 	T)\right)}{\sum_{E\in \cC_{S\backslash T}}z_E}  \nonumber\\
		&= \frac{ \sum_{E\in \cC_{S\backslash T}}z_E\cdot \frac{\partial f(S\backslash( T\cup E))}{\partial S'}+ \sum_{E\in \cC_{S}\backslash \cC_{S\backslash T}}z_E\cdot \frac{\partial f(S\backslash T)}{\partial S'}+ \sum_{\bp\in S\cap S'}\left(f(S\backslash (T\cup \{\bp\})) - f(S\backslash T)\right) }	{\sum_{E \in \cC_S}z_E}\nonumber \\
		\label{eqn:liftderivative_lowerbound}	& \geq \frac{ \sum_{E\in \cC_{S\backslash T}}z_E\cdot \frac{\partial f(S\backslash( T\cup E))}{\partial S'}+ \sum_{E\in \cC_S\backslash \cC_{S\backslash T}}z_E\cdot \frac{\partial f(S\backslash T)}{\partial S'}  }	{\sum_{E\in \cC_S}z_E}.
	\end{align}
	Using \Cref{eqn:pseudo_derivative} and the fact that $T \in \cC_{S-\bp}$ for any $\bp\in S\cap S'$,
	\begin{align*}
		\frac{\hat{\partial} f(S-\bp)}{\hat{\partial} z_T} & = \frac{\sum_{E \in \cC_{S-\bp}}z_E\cdot \frac{\hat{\partial} f(S\backslash( E \cup \{\bp\}))}{\hat{\partial} z_T}+ f(S\backslash (T\cup \{\bp\}))- \sum_{E\supseteq S-\bp }z_E}{\sum_{E\in \cC_{S-\bp}}z_E}  \\
		&= \frac{\sum_{E\in \cC_{S-\bp} }z_E\cdot \frac{\hat{\partial} f(S\backslash( E \cup \{\bp\}))}{\hat{\partial} z_T}+  \sum_{E\in \cC_S\backslash \cC_{S-\bp}}z_E\cdot \frac{\hat{\partial} f(S\backslash \{\bp\})}{\hat{\partial} z_T} + f(S\backslash (T\cup \{\bp\}))- \sum_{E\supseteq S-\bp}z_E}{\sum_{E\in \cC_S}z_E}.
	\end{align*}
	Using \Cref{eqn:pseudo_derivative}, 
	\begin{align*}
		&\frac{\hat{\partial} f(S)}{\hat{\partial} z_T} =	\frac{\sum_{E\in \cC_S }z_E\cdot \frac{\hat{\partial} f(S\backslash E)}{\hat{\partial} z_T}+ f(S\backslash T)- \sum_{E\supseteq S}z_E}{\sum_{E\in \cC_S}z_E}.
	\end{align*}
	The numerator of $	\frac{\hat{\partial} f(S-\bp)}{\hat{\partial} z_T} - \frac{\hat{\partial} f(S)}{\hat{\partial} z_T}$ is 
	\begin{align}
		&\sum_{E\in \cC_{S-\bp}}z_E\left(\frac{\hat{\partial} f(S\backslash( E \cup \{\bp\}))}{\hat{\partial} z_T}-  \frac{\hat{\partial} f(S\backslash E)}{\hat{\partial} z_T}\right) + \sum_{E\in \cC_S\backslash \cC_{S-\bp}}z_E\left(\frac{\hat{\partial} f(S\backslash \{\bp\})}{\hat{\partial} z_T} - \frac{\hat{\partial} f(S\backslash E)}{\hat{\partial} z_T}\right)\nonumber \\ &-\sum_{E : S\backslash E = \{\bp\}}z_E+\left(  f(S\backslash (T\cup \{\bp\}))- f(S\backslash T)\right) \nonumber \\
		\label{eqn: liftderivative_summationlowerbound}	& \geq \sum_{E\in \cC_{S-\bp}}z_E\left(\frac{\hat{\partial} f(S\backslash( E \cup \{\bp\}))}{\hat{\partial} z_T}-  \frac{\hat{\partial} f(S\backslash E)}{\hat{\partial} z_T}\right) + 	\sum_{E\in \cC_S\backslash \cC_{S-\bp}}z_E\left(\frac{\hat{\partial} f(S\backslash \{\bp\})}{\hat{\partial} z_T} - \frac{\hat{\partial} f(S\backslash E)}{\hat{\partial} z_T}\right) \nonumber \\ &-\sum_{E: S\backslash E= \{\bp\}}z_E.
	\end{align}
	Note that the condition $E \in \cC_{S}\backslash \cC_{S-\bp}$ holds only when $\bp\in S$ and $E\cap S$ is either $S\backslash \{\bp\}$ or $\{\bp\}$ from \Cref{obs:boundary_point}. When $E\cap S=\{\bp\}$, the second derivative term in \Cref{eqn: liftderivative_summationlowerbound} is zero. Otherwise, it is $\frac{\hat{\partial} f(S\backslash \{\bp\})}{\hat{\partial} z_T} - \frac{\hat{\partial} f(\{\bp\})}{\hat{\partial} z_T}= \frac{\hat{\partial} f(S\backslash \{\bp\})}{\hat{\partial} z_T}$. Using these observations, \Cref{eqn: liftderivative_summationlowerbound} can be simplified and can be re-written as
	\begin{align}
		\label{eqn:3}&\sum_{E\in \cC_{S-\bp}}z_E\left(\frac{\hat{\partial} f(S\backslash( E \cup \{\bp\}))}{\hat{\partial} z_T}-  \frac{\hat{\partial} f(S\backslash E)}{\hat{\partial} z_T}\right)+ \sum_{E : S\backslash E = \{\bp\}}z_E\left(\frac{\hat{\partial} f(S\backslash \{\bp\})}{\hat{\partial} z_T} -1 \right).
	\end{align}
	The remaining numerator of $\frac{\partial f(S\backslash T)}{\partial S'}$ from \Cref{eqn:liftderivative_lowerbound} is 
	\begin{align}
		\label{eqn:4}\sum_{E\in \cC_{S\backslash T}}z_E\cdot \frac{\partial f(S\backslash( T\cup E))}{\partial S'}+ \sum_{E\in \cC_{S}\backslash \cC_{S\backslash T}}z_{E}\cdot \frac{\partial f(S\backslash T)}{\partial S'}.
	\end{align}
	It remains to prove that the sum of \Cref{eqn:4}, and \Cref{eqn:3} summed over all $\bp\in S\cap S'$ is non-negative. Let us do that by carefully partitioning all the $E$ into groups and do a case analysis. 
	\begin{description}
		\item[Case 1: ($E\in \cC_{S}\backslash \cC_{S\backslash T}$)]
		First, observe that this happens only when either $\emptyset \neq S\cap E \subseteq S\cap T$ or $\emptyset \neq S\backslash E \subseteq S\cap T$ as shown in figures \Cref{fig:case1a} and \Cref{fig:case1b} respectively.
		\begin{figure}[ht]
			\begin{subfigure}{.5\textwidth}
				\centering
				\includegraphics[width=.8\linewidth,page=8]{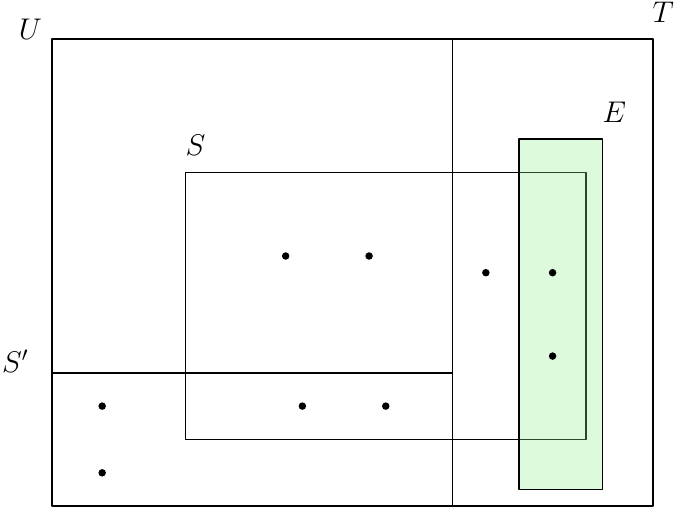}  
				\caption{Case 1a}
				\label{fig:case1a}
			\end{subfigure}
			\begin{subfigure}{.5\textwidth}
				\centering
				\includegraphics[width=.8\linewidth,page=9]{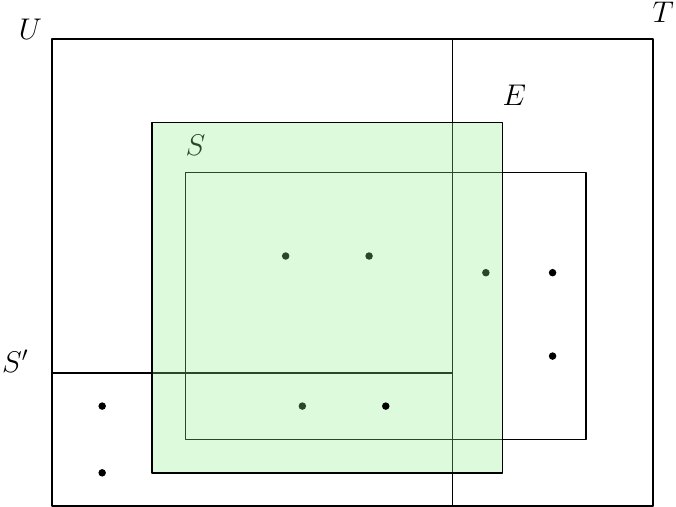}  
				\caption{Case 1b}
				\label{fig:case1b}
			\end{subfigure}
			\caption{Case 1}
			\label{fig:case1}
		\end{figure}
		Observe that in either case, for that $E$, we have $E\in \cC_{S-\bp}$ for any $\bp\in S\cap S'$. 
		\begin{description}
			\item[Case 1a: ($\emptyset \neq S\cap E \subseteq S\cap T$ )]
			In this case, the coefficient of $z_E$ in the sum of \Cref{eqn:4} and, \Cref{eqn:3} summed over all $\bp\in S\cap S'$ is 
			\begin{align}
				\label{eqn:5}
				\frac{\partial f(S\backslash T)}{\partial S'} + \sum_{\bp\in S\cap S'}\left (\frac{\hat{\partial} f(S\backslash( E \cup \{\bp\}))}{\hat{\partial} z_T}-  \frac{\hat{\partial} f(S\backslash E)}{\hat{\partial} z_T} \right).
			\end{align}
			Which we can re-write as 
			\begin{align*}
				\frac{\partial f((S\backslash E)\backslash T)}{\partial S'} + \sum_{\bp_r\in (S\backslash E)\cap S'}\left (\frac{\hat{\partial} f((S\backslash E) \backslash \{\bp\})}{\hat{\partial} z_T}-  \frac{\hat{\partial} f(S\backslash E)}{\hat{\partial} z_T} \right).
			\end{align*}
			This, we can argue is non-negative by induction on $|S|$ by setting $S\leftarrow S\backslash E$. But note that before using inductive hypothesis, we need $T\in \cC_{S\backslash E}$ which is not guaranteed when $E=S\cap T$. But in that case, the terms inside the summation of \Cref{eqn:5} are zero because $T\cap (S\backslash E)=\emptyset $ which implies that the coefficient is  $\frac{\partial f(S\backslash T)}{\partial S'}= \frac{\partial f(S\backslash (T\cup E))}{\partial S'}\geq 0$ from \Cref{lem:lift_derivative}. 
			\item[Case 1b: ($\emptyset \neq S\backslash E \subseteq S\cap T$)]
			In this case, the coefficient of $z_E$ is the same expression as \Cref{eqn:5} but observe that $S\backslash (E\cup \{\bp\})=S\backslash E$ which implies that the coefficient is simply $\frac{\partial f(S\backslash T)}{\partial S'}$ which we know is non-negative from \Cref{lem:lift_derivative}.
			
		\end{description}
		\item[Case 2: ($E \in \cC_{S\backslash T}$)] 
		Let us branch based on whether there exists $\bp\in S\cap S'$ such that $E \notin \cC_{S-\bp}$. This can happen in two ways. First, there exists a point $\bp^{*}\in S\cap S'$ such that $S\backslash E=\{\bp^{*}\}$. Second, there exists a point point $\bp^{*}\in S\cap S'$ such that $S\backslash E= S\backslash \{\bp^{*}\}$
		\begin{figure}[ht]
			\begin{subfigure}{.5\textwidth}
				\centering
				\includegraphics[width=.8\linewidth,page=10]{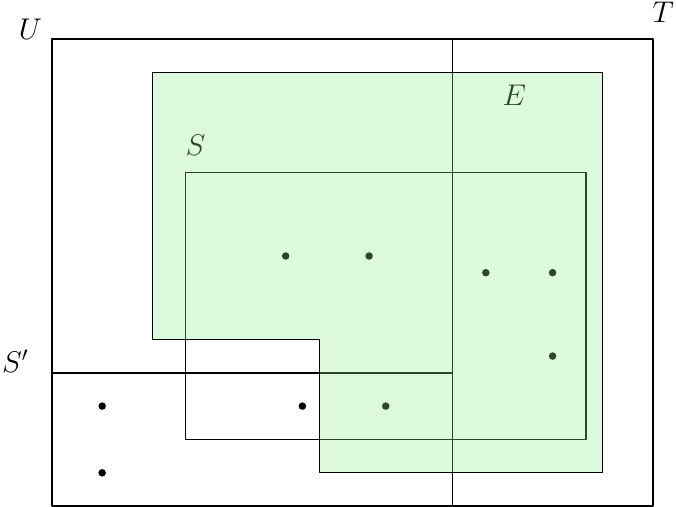}  
				\caption{Case 2a}
				\label{fig:case2a}
			\end{subfigure}
			\begin{subfigure}{.5\textwidth}
				\centering
				\includegraphics[width=.8\linewidth,page=11]{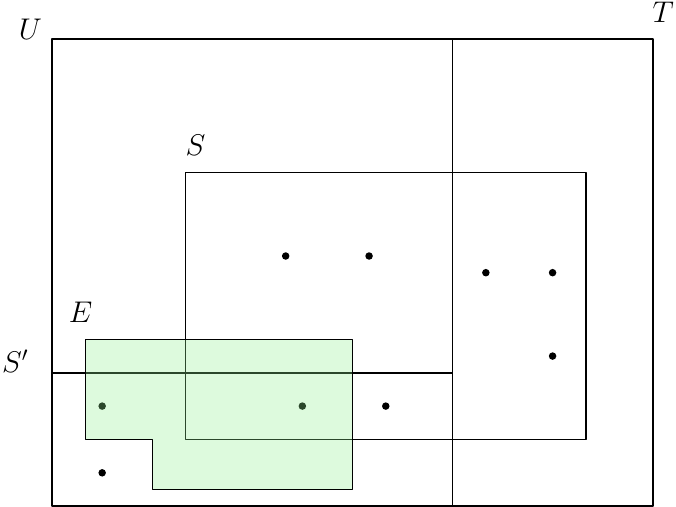}  
				\caption{Case 2b}
				\label{fig:case2b}
			\end{subfigure}
			\caption{Case 2}
			\label{fig:case2}
		\end{figure}
		\begin{description}
			\item[Case 2a: ($\exists \bp^{*}\in S\cap S', S\backslash E=\{\bp^{*}\}$)]
			First, note that $E \in \cC_{S\backslash T}$ for such a cut (see \Cref{fig:case2a}).
			Observe that for any $\bp\neq \bp^{*} \in S\cap S'$, we have $S\backslash (E\cup \{\bp\})= S\backslash E$ and there cannot be a different $\bp$ such that $S\backslash E=\{\bp\}$ for the same $E$. This implies that the coefficient of $z_E$ in the sum of \Cref{eqn:4} and, \Cref{eqn:3} summed over all $\bp\in S\cap S'$ is 
			\begin{align*}
				&\frac{\partial f(S\backslash (T\cup E))}{\partial S'}+ \frac{\hat{\partial} f(S\backslash \{\bp^{*}\})}{\hat{\partial} z_T} -1 \\
				&= \frac{\partial f(\{\bp^{*}\})}{\partial S'}+\frac{\hat{\partial} f(S\backslash \{\bp^{*}\})}{\hat{\partial} z_T} -1 \\
				&= \frac{\hat{\partial} f(S\backslash \{\bp^{*}\})}{\hat{\partial} z_T} \geq 0
			\end{align*}
			\item[Case 2b: ($\exists \bp^{*}\in S\cap S', S\backslash E=S\backslash \{\bp^{*}\}$)]
			First, note that $E \in \cC_{S\backslash T}$ for such a cut (see \Cref{fig:case2b}). Observe that for any $\bp\neq \bp^{*} \in S\cap S'$, we have $E\in \cC_{S-\bp}$. This implies that the coefficient of $z_E$  in the sum of \Cref{eqn:4} and, \Cref{eqn:3} summed over all $\bp\in S\cap S'$ is 
			\begin{align*}
				&\frac{\partial f(S\backslash( T\cup E))}{\partial S'}+ \sum_{\bp\neq \bp^{*} \in S\cap S'}\left(\frac{\hat{\partial} f(S\backslash( E\cup \{\bp\}))}{\hat{\partial} z_T}-  \frac{\hat{\partial} f(S\backslash E)}{\hat{\partial} z_T}\right) \\
				&= \frac{\partial f((S\backslash \{\bp^{*}\}) \backslash T )}{\partial S'}+ \sum_{\bp \in (S\backslash \bp^{*})\cap S'}\left(\frac{\hat{\partial} f((S\backslash \{\bp^{*}\}) \backslash \{\bp\})}{\hat{\partial} z_T}-  \frac{\hat{\partial} f(S\backslash \{\bp^{*}\})}{\hat{\partial} z_T}\right).
			\end{align*}
			This, we can argue is non-negative by induction on $|S|$ ($S\leftarrow S\backslash \{\bp^{*}\}$). Note that $T\in \cC_{S-\bp^{*}}$ holds so we can use induction hypothesis.
			\item[Case 2c: ($E \in \cC_{S-\bp}, \forall \bp \in S\cap S'$ )]
			In this case, the coefficient of $z_E$ is simply,
			\begin{align*}
				& \frac{\partial f(S\backslash( T\cup E))}{\partial S'} + \sum_{\bp\in S\cap S'}\left(\frac{\hat{\partial} f(S\backslash( E \cup \{\bp\}))}{\hat{\partial} z_T}-  \frac{\hat{\partial} f(S\backslash E)}{\hat{\partial} z_T}\right) \\
				&= \frac{\partial f((S\backslash E)\backslash T)}{\partial S'} + \sum_{\bp\in (S\backslash E)\cap S'}\left(\frac{\hat{\partial} f((S\backslash E) \backslash \{\bp\}))}{\hat{\partial} z_T}-  \frac{\hat{\partial} f(S\backslash E)}{\hat{\partial} z_T}\right).
			\end{align*}
			This, we can again argue is non-negative by induction on $|S|$ ($S\leftarrow E$). Note that if $T\notin \cC_{S\backslash E}$, then $E\supseteq S\cap T$ in which case, the terms in the summation are zero and we are done directly without induction. 
		\end{description}
	\end{description}
	This completes the proof of \Cref{lem:final_blow}.
      \end{proof}
We can now wrap up: \Cref{lem:final_blow} was the missing piece in the
proof of \Cref{lem:lift_augmentation}. In turn,
because of \Cref{lem:lift_augmentation} we can assume all points in
$S\backslash T$ to have the same norm. \Cref{lem:special_case}
shows the desired bound of \Cref{lem:final} for this uniform case, completing
the proof.

\appendix

\section{Proofs from \Cref{sec:expon-clocks-proof}}
\label{app:clocks}

\begin{proof}[Proof of~\Cref{lem:single-point}]
  For any client $\bm x \in \bm X$. Since the \RT algorithm is
  translation and scaling invariant, we imagine that $\bm x = \bm
  0$. Now the expected cost incurred by this client is at most the
  distance to the unique center in its leaf region in the tree
  produced by the \RT algorithm, which is at most $\alpha(|\cU|)$. The
  claim now follows by scaling by the true distance to the closest
  center $\| \bm x - \pi(\bm x) \|$, summing over all
  $\bm x \in \bm X$, and using linearity of expectations.
\end{proof}

\subsection{Going from $\ell_1$ Metrics to Cut Metrics}
\label{sec:cut-metrics}

It is known that point set in $\ell_1$ can be written as a
non-negative sum of cut metrics~\cite{DL97}; we give the details here
for completeness. Given a point set $V \in \RR^d$, define
$\ell_i := \min_{\bm v \in V} v_i$ and $u_i = \max_{\bm v \in V}
v_i$. Then $L_1(V) := \sum_{i = 1}^d (u_i - \ell_i)$, and $D_1(V)$ is
the uniform distribution over
$\{(i, \theta) \mid \theta \in [\ell_i, u_i]\}$. Define for any
$S \sse V$ the non-negative quantity
\[ z_S = L_1(V)\cdot \Pr_{(i,\theta) \sim D_1(V)}[ (V \cap \{ \bm x
  \mid \text{sign}(\theta) \cdot x_i \geq \theta \}) = S ] . \] This
is a scaled version of the probability that for a random threshold
cut, the points of $V$ in the halfspace not containing the origin
equals $S$.  A direct calculation shows that for all
$\bm p, \bq \in V$, we have
\[ \| \bm p - \bm q \|_1 = \sum_S z_S \ones_{(|S \cap \{\bm p, \bm
    q\}| = 1)}. \] Now define an $\ell_1$-embedding
$\varphi: V \to \RR^{2^{|V|}}_{\geq 0}$ by setting, for any $S \sse V$,
\[ \varphi(\bm p)_S = \sum_S z_S \ones_{(\bm p \in S)}. \] Again,
$\| \bm p - \bm q \|_1 = \| \varphi(\bm p) - \varphi(\bm q)
\|_1$. Moreover, if the origin belongs to $V$, we get $\varphi(\bm 0)
= \bm 0$. 

\subsection{Proof of \Cref{cor:last-last}}
\label{app:last-last}

Let $\cS = \langle S_{1},\dots,S_{2^k} \rangle$ be the sequence of
cuts in increasing order of their sample values $X_{S}$. (We add the
subsets with $z_S = 0$ at the end of the sequence in some fixed but
arbitrary order.) However, it is not true that we remove points from
$U$ in this order: we need to reject cuts that do not cross the
remaining set $U$. (We say a set $A$ \emph{crosses} $B$ if $ B\cap A $
is a non-empty proper subset of $B$: i.e., if both $B \setminus A$ and
$B \cap A$ are non-empty.) Hence, we recast the ``last-point'' process
again as follows. Given any subset of points $U \sse V$:
\begin{enumerate}
\item Define $U^0 := U$ and $\cS^{0} = \langle\rangle$. In general,
  $U^{r}$ is the set of points remaining after considering sets
  $S_{1},\dots,S_{r}$, and let $\cS^{r}$ is a sequence of cuts
  selected until this point from the sequence~$\cS$. 
\item We define $\cS^{r+1} \gets \cS^{r} \circ \langle S_{r+1}\rangle$ if $S_{r+1}$
  crosses $U^{r}$, else we define $\cS^{r+1} \gets \cS^{r}$. Either
  way, $U^{r+1}= U \backslash \bigcup_{S\in \cS^{r+1}}S$.
\end{enumerate}
Note that $U^r$ and $\cS^r$ are both functions of $(U,\cS)$. Call the
cuts in $\cS^{2^k}$ to be the \emph{valid} cuts for set $U$. 
Given the same sequence of cuts $\cS$ we may get different
subsequences for each subset $U$ of $V$.  So it is not necessarily true
that $(U\backslash T)^r = U^r \backslash T$, because the set of valid
cuts can differ when considering point sets $U$ and $U\backslash
T$. But it turns out that we can still relate $(U\backslash T)^r$ and
$U^r \backslash T $ in some settings.

\begin{lemma}
  \label{lem:p_is_last}
  Given  sequence $\cS$ and index $0\leq r\leq 2^k$ such that
  $U^r \backslash T \neq \emptyset$, we have
  $(U\backslash T)^r=U^r\backslash T$.
\end{lemma}

\begin{proof}[Proof of \Cref{lem:p_is_last}]
  We prove this by induction on $r$. For $r=0$, we know that
  $U^{0}\backslash T=U\backslash T = U^0\backslash T$. Suppose the
  claim holds for $r=t$, then we want to show it holds for $r=t+1$.
  Suppose $U^{t+1} \backslash T \neq \emptyset$, then since
  $U^{t+1} \sse U^t$ we have $U^{t} \backslash T \neq \emptyset$, and
  by the inductive hypothesis we get that
  $(U\backslash T)^t=U^t\backslash T$. In particular, we get that
  $(U\backslash T)^t \sse U^t$. Hence if the new cut $S_{t+1}$ crosses
  $(U\backslash T)^t$, then it also crosses $U^t$, and therefore
  \[ (U\backslash T)^{t+1}= (U\backslash T)^{t} \setminus S_{t+1} =
    (U^t\backslash T) \setminus S_{t+1} = (U^{t+1}\backslash T). \] So
  suppose the new cut $S_{t+1}$ does not cross $(U\backslash T)^t$,
  and thus $(U\backslash T)^{t+1}= (U\backslash T)^{t}$. There are two
  cases: either $(U\backslash T)^t \subseteq S_{t+1}$ or
  $(U\backslash T)^t \cap S_{t+1}=\emptyset$. In the first case, if
  $S_{t+1}$ crosses $U^t$, then
  $U^{t+1}\backslash T = (U^t \setminus S_{t+1}) \setminus T =
  \emptyset$, and hence there is nothing to prove. Else if $S_{t+1}$
  does not cross $U^t$, then $U^{t+1} = U^t$ and also
  $(U\backslash T)^{t+1}= (U\backslash T)^{t}$, so we are done using
  the inductive hypothesis.

  In the second case,
  $(U\backslash T)^{t}\cap S_{t+1}= (U^{t}\backslash T)\cap
  S_{t+1}=\emptyset$, so we get that $(U^t \cap S_{t+1}) \subseteq T$.
  This means that $U^{t+1}\backslash T = U^{t}\backslash T $ regardless
  of whether $S_{t+1}$ crosses $U^t$. Since $U^{t}\backslash T
  =(U\backslash T)^{t} = (U\backslash T)^{t+1}$, we are done.
\end{proof}

As discussed above, $|U^{2^k}| = 1$, and we define this unique point
$p \in U^{2^k}$ to be the ``last'' point in $U$, and we call this
event ``$p$ is last in $U$''.

\LastPointLemma*

\begin{proof}
Using the definition of the event \textnormal{``$p$ is last in $U$''}, we know that $U^{2^{k}}=\{p\}$ and since $p\notin T$, we have $U^{2^k}\backslash T \neq \emptyset$. Using \Cref{lem:p_is_last}, we can say that $(U\backslash T)^{2^k}=\{p\}$.  
\end{proof}

\section{Proofs from Section~\ref{sec:tight-lower-bounds}}

\SetSystem*

\begin{proof}[Proof of~\Cref{lem:random-construction}]
  For some parameter $p \in (0,\nf12)$, %
  consider $k$ independently chosen sets, each obtained by adding in
  each element of $[k]$ independently with probability $p$. The
  expected size of each such set is $\bar{s} := pk$; moreover, each
  element of $[k]$ should hit a $p$ fraction of the sets, so hitting
  all the $k$ sets should require
  $\bar{h} := \ln_{1-p}(1/k) \approx (1/p) \ln k$ sets, giving
  $\bar{s} \bar{h}/k = \ln k$. We now show that with non-zero
  probability, there does exist a set system with parameters close to
  these.

  Define $\cB_1$ to be the bad event that some set $S_i$ has size
  smaller than $s := (1-\eps)pk$, and $\cB_2$ to be the bad event that the
  hitting set has size at most some parameter $h$. We now show that
  for suitable choices of $\eps$ and $h$, we have
  $\Pr[ \cB_1 ] < \nf12$ and $\Pr[ \cB_2 ] \leq \nf12$, which
  completes the proof. We consider the event $\cB_2$ first: a union
  bound shows that
  \[ \Pr[ \cB_2 ] \leq \sum_{H: |H| = h} \Pr\left[ \forall i \in [k] : S_i
    \cap H \neq \emptyset \right] = {k \choose h} (1-(1-p)^h)^k \leq
    \frac{(2k)^h}{2} \cdot e^{-k(1-p)^h}. \] Setting this upper bound
  to equal $\nf12$, we get
  \begin{gather}
    \ln h-h\ln(1-p) = \ln\left( \frac{k}{\ln 2k}\right).
  \end{gather}
  We now use that $-p-p^2 \leq \ln (1-p)\leq -p$ for $p\in
  [0,\nf12]$ to get
  \begin{gather}
    \ln h +hp\leq \ln \left(\frac{k}{\ln 2k} \right)\leq \ln h
    +hp(1+p). \label{eq:4}
  \end{gather}
  Since $h \geq 1$, the left-most inequality of~(\ref{eq:4}) gives
  $h \leq (1/p) \ln (\frac{k}{\ln 2k})$. However, we want a lower
  bound on $h$, so we  substitute this into the
  right-most inequality of~(\ref{eq:4}) to get
  \begin{gather}
    \frac1{1+p} \cdot \ln \underbrace{\left(\frac{k}{\ln 2k} - \frac1p \ln
      \bigg(\frac{k}{\ln 2k}\bigg) \right)}_{(\star)}\leq p\,h = \frac{s\,h}{(1-\eps)k}.
  \end{gather}
  We can now set $p := \frac{2 \ln 2k}{k^{1/3}}$ (which ensures that the second
  term in $(\star)$ is at most half the first for a large enough $k$) and get 
  \[ sh/k \geq \frac{1-\eps}{1+p} \cdot \ln \left(\frac{k}{2\ln 2k}
    \right). \]
  Now setting $\eps := 1/k^{1/3}$ and using a Chernoff bound and a union bound,
  \[ \Pr[ \cB_1 ] < k \cdot e^{-\eps^2 pk/2} = \frac12. \] Taking a
  union bound over the two bad events, we get that with non-zero
  probability our sets in $\cS$ are of size $\approx k^{2/3} \ln
  k$, the hitting set is of size $\approx k^{1/3}$, and $hs/k \geq (1-
  \frac{2\ln 2k}{k^{2/3}}) (\ln k - O(\ln \ln k))$.  
\end{proof}

\section{Proofs from~\Cref{sec:kmeans}}
\label{sec:app-kmeans}

\subsection{Proof of the Stretch-vs.-Separation \Cref{clm:stretch-separation}}
\label{sec:proof-stretch-vs-separation}

\begin{wrapfigure}{R}{0.5\textwidth}
  \centering
  \includegraphics[width=0.7\linewidth,page=16]{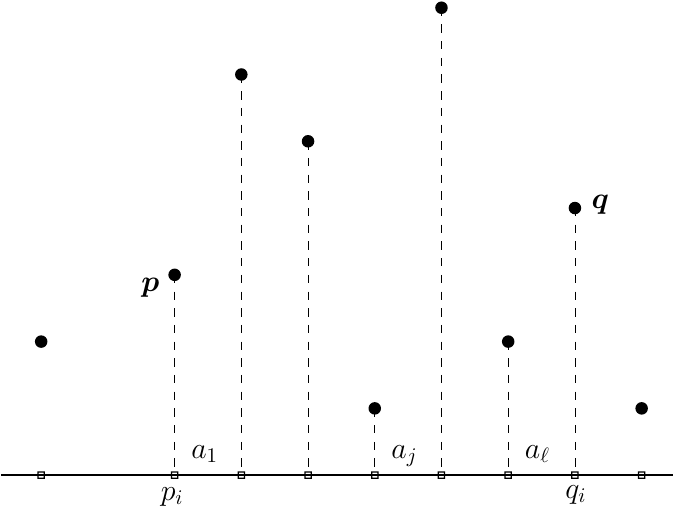}
  \caption{projection of points onto an axis}
  \label{fig:proj}
\end{wrapfigure}
We prove a lemma about point sets that immediately implies
\Cref{clm:stretch-separation}. Consider a set $S \sse \RR^d$ of
points, and focus on $p,q \in S$. Consider some dimension $i$, and
consider the projection of the points onto that dimension (as in the
figure). Let $a_1,\dots,a_\ell$ be the lengths of intervals into which
the line segment joining $p_i$ and $q_i$ is partitioned by projections
of other points in $S$ onto the $i^{th}$ dimension. Any cut
$(i,\theta)$ intersecting the $j^{th}$ interval splits the centers
into two groups with at least $\min(j,\ell+1-j)$ centers on either
side.

Define the \emph{stretch}
$s_i$ between $p$ and $q$ in the $i^{th}$ dimension, and the
\emph{expected separation} ${\text{sep}}_i$ after choosing a random cut that cuts
the $j^{th}$ interval with probability proportional to $a_j^2$ as
\begin{align*}
s_i:= \frac{\big(\sum_{j\in [\ell]}a_j\big)^2}{\sum_{j\in
  [\ell]}a_j^2} \qquad\text{and}\qquad 
\text{sep}_i:=  \frac{\sum_{j\in [\ell]}a_j^2\cdot \min(j,\ell+1-j)}{\sum_{i\in [\ell]}a_j^2}.
\end{align*}

\begin{lemma}
  \label{lem:expected_separation}
  $\text{sep}_i\geq \frac{s_i}{8(1+\ln(\nicefrac{(2|S|)}{s_i}) )}$.
\end{lemma}
Before we prove this, let us generalize this to higher dimensions:
\begin{cor}
  \label{cor:stretch-separation}
  Consider a set $S \sse \RR^d$ and two points $p,q \in S$ having
  stretch $s$. If we choose a threshold cut $(i,\theta)$ from the
  distribution $D_2(S)$ and condition on separating $p,q$, the expected
  number of points in each side of the cut decreases by at least
  $\frac{s}{8(1+\ln (\nicefrac{(2|S|)}{s}))}$.
\end{cor}

\begin{proof}
  The stretch between $p,q$, and the expected separation conditioned
  on separating the two centers is
  \begin{align*}
    \overline{s}:= \frac{\sum_{i\in [d]}\big( \sum_{j\in
    [\ell_i]}a_{i,j}\big)^2}{\sum_{i\in [d]}\sum_{j\in
    [\ell_i]}a_{i,j}^2} \qquad \text{and} \qquad  \overline{sep}:= \frac{\sum_{i\in [d]}\big( \sum_{j\in
    [\ell_i]}a_{i,j}^2\cdot \min(j,\ell_i+1-j)\big)}{\sum_{i\in
    [d]}\sum_{j\in [\ell_i]}a_{i,j}^2},
  \end{align*}
  where $a_{i,j}$ is the width in the partition defined above along
  dimension $i$.  Define a random variable $I \in [d]$ on the
  dimensions, that takes on value $i$ with probability
  \[ \frac{ \sum_{j\in [\ell_i]}a_{i,j}^2}{\sum_{i\in [d]}\sum_{j\in
        [\ell_i]}a_{i,j}^2}. \] Then we have $\overline{s}=\EE_I[s_I]$ and
  $\overline{\text{sep}}=\EE_I[\text{sep}_I]$. Finally, the function
  $h(x)=\frac{x}{1+\ln(\alpha/x))}$ being convex for
  $0\leq x\leq \alpha$, we can use Jensen's inequality to get
  \[
    \overline{\text{sep}}= \mathbb{E}_I[\text{sep}_I] \stackrel{(\text{Lem.\ref{lem:expected_separation}})}{\geq}
    \mathbb{E}_I\left[\frac{s_I}{8\left(1+\ln\left(\frac{2|S|}{s_I}\right)\right)}\right]
    \geq \frac{\overline{s}}{8\left(1+\ln\left(\frac{|S|}{\overline{s}}\right)\right)}. \qedhere
  \]
\end{proof}
Finally, translating to the language of \S\ref{sec:kmeans} shows that \Cref{clm:stretch-separation} is just a reformulation of
\Cref{cor:stretch-separation}. So it suffices to prove
\Cref{lem:expected_separation}, which we do next.

\begin{proof}[Proof of~\Cref{lem:expected_separation}]
  Let us look at the following constrained minimization problem
  \begin{align}
    \min \text{sep}_i &= \textstyle \sum_j a_j^2 \cdot \min(j,\ell+1-j) \label{eqn:primal} \\ 
    s.t.\quad &\textstyle \sum_j a_j = \sqrt{s_i} \nonumber \\  
             &\textstyle \sum_{j} a_j^2 = 1.  \nonumber
  \end{align}
  The Lagrangian dual of the above primal program is
  \begin{align*}
    \mathcal{L}(a,\lambda,\gamma) = \sum_{j}a_{j}^2\cdot
    \min(j,\ell+1-j) -2\lambda \bigg( \sum_{j}a_j-\sqrt{s_i}
    \bigg)+\gamma \bigg(\sum_{j}a_{j}^2-1 \bigg), 
  \end{align*}
  and setting the gradient to zero means the minima for $a$ occur when
  $a_{j} = \frac{\lambda}{\gamma+\min(j,\ell+1-j)}$. (We assume
  $\gamma\geq 0$.) Substituting and simplifying gives
  \begin{align*}
    2\lambda\sqrt{s_i}-\sum_{j}\frac{\lambda^2}{\gamma+\min(j,\ell+1-j)}-\gamma . 
  \end{align*}
  We can maximize over $\lambda$ which happens when 
  $\sum_{j}\frac{\lambda}{\gamma+\min(j,\ell+1-j)}=\sqrt{s_i}$; substituting gives
  \begin{align*}
    \frac{s_i}{\sum_{j}\frac{1}{\gamma+\min(j,\ell+1-j)}}-\gamma \geq
    \frac{s_i}{2\ln(1+\frac{\ell+1}{2\gamma})}-\gamma  \geq
    \frac{1}{2}\bigg(\frac{s_i}{\ln(1+\frac{|S|}{2\gamma})}-2\gamma
    \bigg). 
  \end{align*}
  It remains to choose $\gamma$. For convenience, we set the above
  expression to $\gamma$; this means 
  \begin{align*}
    \frac{s_i}{\ln(1+\frac{|S|}{2\gamma})}=4\gamma \qquad \iff \qquad
    \frac{2|S|}{s_i}=\frac{|S|/2\gamma}{\ln(1+\frac{|S|}{2\gamma})}. 
  \end{align*}
  Using \Cref{obs:x=ylogy} below, we get $\gamma \geq \frac{s_i}{8(1+\ln(\nicefrac{(2|S|)}{s_i}) )}$.
\end{proof}

\begin{obs}
  \label{obs:x=ylogy}
  For any $x \geq 0$ and $y \geq 1$ such that $y=\frac{x}{\ln(1+x)}$, we have $x\leq 2y(1+\ln y)$. 
\end{obs}
\begin{proof}
  Since $x/\ln(1+x)$ is an increasing function, it is sufficient to prove that 
  \begin{align*}
    y\leq \frac{2y(1+\ln y)}{\ln(1+2y(1+\ln y))}.
  \end{align*}
  Finally, taking derivatives shows $\ln(1+2y(1+\ln y))\leq 2(1+\ln
  y)$ for $y \geq 1$. %
\end{proof}

\section{Proofs from~\Cref{sec:tight}}
\label{sec:app-tight}

\begin{proof}[Proof of  \Cref{lem:reduce-to-deriv}]
  For a point $\bp\in U$, define
  \begin{align}
    \label{eqn:trajectory}
    z_E(t):= \begin{cases} 
      z_E\cdot t & \bp \in E \\
      z_E & \bp \notin E.
    \end{cases}
  \end{align}
  Let $f_t(U)$ and $\ell_t(\bp)$ be the function value of points in
  $U$ and norm of $\bp$ at time $t$ when the embedding is changing
  according to the trajectory given by \Cref{eqn:trajectory}.

  We first claim that $\ell_0(\bp)=0$ and $f_0(U)=0$.  Indeed, We know
  that
  $\ell_t(\bp)=\sum_{S:\bp\in S}z_S(t)=\sum_{S:\bp\in S}z_S\cdot t
  =t\cdot\ell(\bp)$, which implies that $\ell_0(\bp)=0$. The second
  part is proved by induction on $|U|$: If $|U|=1$, then
  $f_0(U)=\ell_0(\bp)=0$. If $|U|\geq 2$ and $z_E(0)=0$ for all
  $E\in \cC_{\Univ}$, then $f_0(U)=0$ by definition. Otherwise, using
  \Cref{eqn:definition_global}, we can write
  \begin{align*}
    f_{0}(U)= \frac{\sum_{E\in \cC_{\Univ}}z_E(0)\cdot f_0(U\backslash U_E )}{\sum_{E\in \cC_{\Univ}}z_E(0)}. 
  \end{align*}
  In the numerator of $f_0(U)$, either $z_E(0)=0$ when $\bp\in E$,
  or $f_0(U\backslash U_E)=0$ when $\bp\in U\backslash U_E$ by
  inductive hypothesis which concludes the proof.

  Now using the chain rule and the assumed bound on the derivative,
  \begin{align*}
    \frac{\mathrm{d}f_t(U)}{\mathrm{d}t}=
    \sum_{E\subseteq [k]}\frac{\partial f(U)}{\partial z_E}\Big|_{\bm z=\bm z(t)}\cdot
    \frac{\mathrm{d}z_E(t)}{\mathrm{d}t}=& \sum_{E: \bp\in E}z_E\cdot \frac{\partial f(U)}{\partial z_E}\Big|_{\bm z=\bm z(t)}\\
                                         &\leq \sum_{E: \bp\in E}z_E\cdot \beta_{k-1}=\beta_{k-1}\cdot \ell(\bp).
  \end{align*}
  Integrating gives $f_1(U)-f_0(U)=f(U)-f_0(U) = f(U) \leq \beta_{k-1}\cdot \ell(\bp)$. 
\end{proof}

\begin{proof}[Proof of~\Cref{lem:derivative_observations}]
  The proof is by induction on $|S|$.
	The base case is when $|S|=1$. If $S=\{r\}$, then we have $\frac{\partial f(S)}{\partial z_T}=\frac{\partial \ell(r)}{\partial z_T}$. We know that $\ell(r)$ can be written as 
	\begin{align}
		\label{eqn:norm_expansion}
		\ell(r)= \sum_{E\ni r}z_E.
	\end{align}
	The derivative $\frac{\partial \ell(r)}{\partial z_T}$ is equal to $1$ if $z_T$ appears as a term in the expansion of $\ell(r)$ as in \Cref{eqn:norm_expansion} and is equal to $0$ otherwise. In the case of statement \ref{lem:shift_derivative}, since $T\in \cC_S^{\ell}$, we have $S\subseteq T \implies r\in T$ this concludes the base case for statement \ref{lem:shift_derivative}. Similarly, if statement \ref{lem:empty_derivative} holds, $z_T$ does not appear as a term in the expansion of $\ell(r)$ and hence the partial derivative of $f(S)=\ell(\mu_r)$ with respect to $z_T$ is zero.
	
	For $|S|\geq 2$, we have $\mathbbm{1}[T\in \cC_S]=0$ in both the cases of statements \ref{lem:shift_derivative} and \ref{lem:empty_derivative}. Using this in \Cref{eqn:derivative} gives
	\begin{align}
		\label{eqn:derivative_boundarycase}
		\frac{\partial f(S)}{\partial z_T} = \frac{\sum_{E\in \cC_S}z_E\cdot \frac{\partial f(S\backslash E)}{\partial z_T}}{\sum_{E\in \cC_S}z_E}. 
	\end{align}
	For the inductive step, since
    $S\subseteq T$ in statement \ref{lem:shift_derivative}, we have $(S\backslash E) \subseteq T $. By the inductive
	hypothesis, the recursive derivative terms
	$\frac{\partial f(S\backslash E)}{\partial z_T}$ are
	equal to $1$. Using \Cref{eqn:derivative_boundarycase}, we are
	done. The inductive step for the case of statement
	\ref{lem:empty_derivative} follows similarly where all terms $\frac{\partial f(S\backslash E)}{\partial z_T}$ are equal to $0$. Finally, using the fact that
	$f(S)$ is always non-negative for any $z$, and the fact that
	$\frac{\partial f(S)}{\partial z_E}=1$ for any
	$S\subseteq E $, we have
	$f(S)\geq \sum_{E\supseteq S}z_E$.
\end{proof}

\begin{proof}[Proof of~\Cref{lem:pseudoderivative_upperbound}]
  The proof is by induction on $|S|$. Observe that for
  $T\notin \cC_S$, either $T\supseteq S$ or $T\cap S=\emptyset $, and then $\frac{\widehat{\partial}f(S)}{\widehat{\partial}z_T}=\frac{\partial f(S)}{\partial z_T}$ by \Cref{defn:pseudo_derivative} and \Cref{lem:derivative_observations}. The fact that the pseudo-derivative is non-negative in this case is immediate from the \Cref{defn:pseudo_derivative}. For $|S|=1$,
  since $T\notin \cC_S$, we are done. If
  $|S|\geq 2$ and $T\in \cC_S$, the inductive hypothesis implies that $\frac{\widehat{\partial}f(S\setminus E)}{\widehat{\partial}z_T}\geq \max \left(\frac{\partial f(S\setminus E)}{\partial z_T},0 \right)$. It remains to prove that
  \begin{align*}
    \max\left(f(S\setminus T)-f(S),0  \right)\leq  f(S\setminus T)-\sum_{E\supseteq S }z_E.
  \end{align*}
  Using \Cref{lem:derivative_observations}, we get $f(S\setminus T)-f(S)\leq f(S\setminus T)-\sum_{E\supseteq S}z_E$. The other inequality $\sum_{E\supseteq S }z_E \leq f(S\setminus T)$ follows from the below argument, again using statement \ref{cor:shift_inequality} of \Cref{lem:derivative_observations}:
  \[
    f(S\setminus T)\geq \sum_{E \supseteq S\setminus T }z_E\geq
    \sum_{E\supseteq S}z_E \]
  This concludes the proof of the lemma.
\end{proof}

{\small
\bibliographystyle{alpha}
\bibliography{Writeup/bibdb}
}

\end{document}